\documentclass[12pt]{article} 
\usepackage[sectionbib]{natbib}
\usepackage{graphicx}
\usepackage{amsmath}
\usepackage{amssymb}
\usepackage{amsfonts}
\usepackage{adjustbox}
\usepackage{multirow}
\usepackage{amsthm}
\usepackage{graphicx}
\usepackage{bbm}
\usepackage{algorithm}
\usepackage{adjustbox}
\usepackage[]{hyperref}
\usepackage{algorithm}
\usepackage{algpseudocode}
\newtheorem{asmtn}{Assumption}
\newtheorem{theorem}{Theorem}
\newtheorem{lemma}{Lemma}
\newtheorem{corollary}{corollary}

\theoremstyle{definition}

\def\mb{\mathbf}
\title{Large-scale adaptive multiple testing for sequential data controlling false discovery and nondiscovery rates}
\author{Rahul Roy, Shyamal K. De, Subir Kumar Bhandari}
\date{}
\begin{document}
	\maketitle
	\begin{abstract}
   In modern scientific experiments, we frequently encounter data that have large dimensions, and in some experiments, such high dimensional data arrive sequentially or in stages rather than full data being available all at a time. We develop multiple testing procedures with simultaneous control of false discovery and nondiscovery rates when $m$-variate data vectors $\mb{X}_1, \mb{X}_2, \dots$ are observed sequentially or in groups and each coordinate of these vectors leads to a hypothesis testing. Existing multiple testing methods for sequential data uses fixed stopping boundaries that do not depend on sample size, and hence, are quite conservative when the number of hypotheses $m$ is large. We propose sequential tests based on adaptive stopping boundaries that ensure shrinkage of the continue sampling region as the sample size increases. Under minimal assumptions on the data sequence, we first develop a test, i.e., a stopping and a decision rule, based on an oracle test statistic such that both false discovery rate (FDR) and false nondiscovery rate (FNR) are nearly equal to some prefixed levels with strong control at these levels. Under a two-group mixture model assumption, we propose a data-driven stopping and decision rule based on local false discovery rate statistic that mimics the oracle rule and guarantees simultaneous control of FDR and FNR asymptotically as $m$ tends to infinity. Both the oracle and the data-driven stopping times are shown to be finite (i.e., proper) with probability 1 for all finite $m$ and converge to a finite constant as $m$ grows to infinity. Further, we compare the data-driven test with the existing ``gap'' rule proposed in \cite{HB21} and show that the ratio of the expected sample sizes of our method and the ``gap'' rule tends to zero as $m$ goes to infinity. Extensive analysis of simulated datasets as well as some real datasets illustrate the superiority of the proposed tests over some existing methods.
	\end{abstract}
 
		\vspace{9pt}

				\section{Introduction} 
    
			Technological advancements in the past few decades have challenged statisticians with the task of drawing inferences for high-dimensional or ultrahigh-dimensional data which led to the development of a myriad of innovative statistical methodologies suited for those settings. Among many high-dimensional statistical problems, large-scale multiple testing got much attention due to its vast applications in many areas including genetics, medical imaging, and astrophysics, to name a few. For instance,``high throughput'' devices such as microarrays produce gene expression levels and a statistician needs to compare sick and healthy subjects for thousands or even millions of genes simultaneously (\cite{dudoit2008multiple}). Multiple testing of a large number of hypotheses also arises in many spatial settings such as disease mapping (\cite{green2002hidden}), astronomical surveys (\cite{miller2001controlling}), and public health surveillance (\cite{caldas2006controlling}), among others. 
A widely popular error metric for large-scale multiple testing is the false discovery rate (FDR), proposed in the seminar paper by \cite{BH95}, which is the expected value of false discovery proportion (FDP) defined as the ratio of the number of false rejections divided by the total number of rejections.
Substantially large number of statistical methods have been developed to control FDR or some variations of it such as marginal FDR (mFDR, \cite{GL02}), positive FDR (pFDR, \cite{Storey02}) and false discovery exceedance (FDX, \cite{perone2004false}) and some of these methodologies are developed in the context of large-scale multiple testing when $m$ grows to infinity. To review the literature on large-scale multiple testing and FDR controlling methodologies, we refer to \cite{cai2011compound}, \cite{cao2022optimal}, \cite{SC07}, \cite{Sea15} and the references therein.  

 The above multiple testing procedures are applicable when full data is available at the time of testing and these are known as fixed-sample-size tests. However, there are a number of applications when, instead of full data being available at the time of analysis, data arrive sequentially one at a time or in groups or stages. A statistician may need to draw inferences each time new data arrives based on the data available till that time point. For such sequentially observed data, sometimes referred to as streaming data, one may need to test multiple hypotheses simultaneously. A very natural application of sequential multiple testing is in clinical trials with multiple endpoints where patients are collected sequentially or in groups. At each interim stage of such trials, a decision is made whether to accept or reject or to collect more samples. 
 
			Sequential multiple testing has evolved into its current form in the last decade. Three approaches of sequential testing have been adopted by statisticians. In the first approach, decisions for all the hypotheses are made simultaneously. This approach is favorable when the data arrive as units or vectors of size $m$ and therefore the cost of observing individual coordinates corresponding to different data-streams for each unit is insignificant in respect to that of engaging a new unit. The authors who followed this approach include \cite{DB12a,DB12b,DB15,SF17,SF19,HB21}. 
			
			In the second approach, each data-stream is considered to be sequential results of conducting a different experiment. Therefore it is permissible for data-streams to be stopped at different time points. This approach is helpful when the cost of conducting different experiments is higher than involving a new unit. Related works include \cite{BS14,BS20,Bar18}.
			
			The third approach is known as online multiple testing. Here, instead of the sequential arrival of new observations, new hypotheses are taken into consideration and the objective is to control some measure of false positives like $FDR$ (\cite{FS08,AR14,JM18,RYWJ17,RZWJ18,GSW21}) and $FWER$ (\cite{TR21}). 
			
    In this article, we are following the first approach. i.e., the number of hypotheses($m$) is fixed and new observations appear as vectors of size $m$. In the literature, among the first and second approaches, the following pairs of error matrices were considered to be controlled: $FWER$-$I$ and  $FWER$-$II$ (\cite{DB12a,DB12b,BS14,SF17}); $FDR$ and $FNR$ (\cite{BS20,HB21}); and generalized $FWER$ pair (\cite{DB15,SF19}). These existing methods perform well when the number of hypotheses is small. But in a large-scale setting, the controls become conservative. i.e., as $m$ grows larger, the attained error matrices become much smaller than the supposed control which implies a large stopping sample size. In the non-sequential version of multiple testing, the articles \cite{SC07,SC09,CS09,Sea15} provide a line of optimal methods that work well in large-scale scenarios. In this article, we have followed their methods for obtaining a sequential multiple testing method appropriate for a large number of hypotheses. The main contributions of this paper are the following. I. We provide an oracle rule for a general setup with adaptive boundaries that is proper and achieves $FDR$ and $FNR$ control under minimal conditions. II. Under the two-group mixture model, we develop an adaptive data-driven rule that achieves asymptotic control of $FDR$ and $FNR$ as the number of hypotheses $m$ grows to infinity. III. The stopping time for the oracle and the data-driven rule under the two-group mixture model converges weakly to a finite natural number as $m$ tends to infinity. IV. Finally, we show that the Asymptotic relative efficiency of the `GAP' rule compared to our oracle rule under the two-group mixture model tends to zero. The article is organized in the following way: in section 2 we have described the basic setup and required assumptions. In section 3 the oracle rule for general setup has been proposed and the Theorem \ref{thm_1} is expressed. Section 4 considers a more specialized two-group mixture model set up and the updated oracle rule for this setup is formulated. Also in section 4 a data-driven rule has been proposed and Theorems \ref{thm_2} and \ref{thm_3} have been stated. Section 5 establishes the asymptotic dominance of stopping times of the asymptotically optimal rule of \cite{HB21} on our oracle rule. Section 6 verifies the results stated in this article numerically with simulation studies. In section 7 we have applied our method on two real data sets. Proofs of the theorems and the lemmas are included in the appendices \ref{appendixA}-\ref{appendixD}.

				\section{Model Description} 
			
			Suppose, with respect to some measurable space $(\Omega,\sigma,\mu)$, we are required to test $m$ pairs of hypotheses simultaneously. A vector of binary ($0,1$) random variables $\boldsymbol{\theta}=(\theta_1,\theta_2,\cdots,\theta_m)$ determines the true states of the hypotheses such that, if $\theta_i$ is $0$($1$), then the $i$-th null (alternative) hypotheses is true. A sequence of random $m$-vectors $\mathbf{S_n}= (S_n^1,S_n^2,\cdots,S_n^m)$ are defined on a sequence of monotone increasing sigma fields $\{\sigma_n| n \in \mathbb{N} \}$ such that $\sigma(\cup_{n=1}^\infty \sigma_n)\subseteq\sigma $. $\{S_n^i\}$ is called the $i$-th test statistic at `time' $n$ and is used to test for the $i$-th pair of hypotheses. 
			
			A sequential test for testing the $m$ hypotheses is defined as the pair $(T,\boldsymbol{\delta})$; Where $T$ is a $\{\sigma_n\}$-stopping time (i.e., if $\mathbf{S_n}$ is observed up to time point $T\in \mathbb{N}$, then $T$ is defined on the $\sigma$-field $\sigma_n$). $\boldsymbol{\delta}=(\delta^1,\delta^2,\cdots,\delta^m)\in \{0,1\}^m$ is a decision vector defined on $\sigma$-field $\sigma_T$. We say we reject $i$-th null hypothesis (or simply reject) if $\delta^i=1$ and accept otherwise. The dependency of $\boldsymbol{\delta}$ on $T$ is suppressed for simplicity of notation. Any sequential test $(T,\boldsymbol{\delta})$ defined in such a way makes $V=\sum_{i=1}^m(1-\theta^i)\delta^i$ false rejections among $R=\sum_{i=1}^m \delta^i$ rejections and $W=\sum_{i=1}^m\theta^i(1-\delta^i)$ false acceptance among $m-R=\sum_{i=1}^m(1-\delta^i)$ acceptance when the true states of the hypotheses are given by $\mathbf{\theta}$. We follow a Bayesian path by considering $\mathbf{\theta}$ to be random with some distributional assumption on it. In light of the above discussion, $FDR$ and $FNR$ for the sequential multiple test $(T,\boldsymbol{\delta})$ are defined as follows:
			
			\begin{align}
				\label{fdrd}FDR = &E\bigg( \frac{\sum_{i=1}^m(1-\theta^i)\delta^i}{(\sum_{i=1}^m \delta^i)\vee1}\bigg) \\
				\label{fnrd}	FNR = & E \bigg(\frac{\sum_{i=1}^m(1-\delta^i)\theta^i}{(\sum_{i=1}^m (1-\delta^i))\vee1} \bigg) 
			\end{align}
			
			For simplicity, we have omitted the dependence of the test $(T,\boldsymbol{\delta})$ in the definition of $FDR$ and $FNR$. 
			Fix $\alpha,\beta \in (0,1)$. Our objective is to find a test $(T,\boldsymbol{\delta})$ which controls $FDR$ and $FNR$ at levels $\alpha$ and $\beta$ respectively. Keeping this objective in mind we define a class of tests:
			
			\begin{equation*}
				\Delta(\alpha,\beta)= \{(T,\boldsymbol{\delta}): FDR\leq\alpha  \ \&  \ FNR\leq\beta  \}
			\end{equation*}
			
			A desirable criterion for any sequential simultaneous test $(T,\boldsymbol{\delta})$ is to belong in the class $\Delta(\alpha,\beta)$. In literature, among the existing sequential multiple testing methods, \cite{HB21} and \cite{BS20} provide sequential multiple testing rules $(T,\mathbf{\delta}\in\Delta(\alpha,\beta))$. The premise in \cite{BS20} is different from us as different coordinates are allowed to have different stopping times. The asymptotically optimal decision rule in \cite{HB21} is more suitable for $FWER$ type $I$ and $II$ control and therefore results in a conservative control for $FDR$ and $FNR$. For large-scale multiple testing problems, such rules are not ideal. Along the lines of studies by Sun et al. for multiple testing problems using local fdrs  
			
			\cite{SC09} defined Local Index of Significance(LIS) for $i$ as:
			\begin{equation}
				\label{tni}
				t_n^{*i}=\mathbb{P}(\theta^i=0|\mathbf{S_n}) ; \ i=1(1)m
			\end{equation} 
			Suppose, given $\theta_i=j$, the joint distribution of $\mathbf{S_n}$ be $f^*_n(\mathbf{S_n}|\theta_i=j); \ j=0,1$. So, 
			\begin{equation}
				\label{tndef}
				t_n^{*i}=\frac{\mathbb{P}(\theta_i=0)f^*_n(\mathbf{S_n}|\theta_i=0)}{\mathbb{P}(\theta_i=0)f^*_n(\mathbf{S_n}|\theta_i=0)+\mathbb{P}(\theta_i=1)f^*_n(\mathbf{S_n}|\theta_i=1)}
			\end{equation}
			Although, $LIS$ statistics are generally defined for continuous random variables, this definition can also be used for discrete test statistics. In that case, $f^*_n(\mathbf{S_n}|\theta_i=j)$ is discrete. 
			
			Jointly $t_n^{*1},t_n^{*2},\cdots,t_n^{*m}$ will be used to test the hypotheses
			\begin{equation*}
				H_{0i}: \theta_i = 0 \quad \text{vs. } \quad  H_{1i}: \theta_i = 1 
			\end{equation*}

			A large number of tests can be generalized in this form. For example, suppose, we want to compare the means of some variable of two homogeneous groups where the observations from new units of each group are obtained sequentially. A sequential two-sample t-test is appropriate in such a situation. \cite{CG84} provides a brief review on this topic. Here, if the size of the observations from the $j$-th group is $n_j$, $j\in\{1,2\}$, we can consider the sequential $t$ statistics with total sample size $n (= n_1 + n_2)$ to be $S_n$ and thus our setup applies here.   
			
			In the multiple testing literature, authors have mostly used the p values as the test statistics. For the sequential cases too, we can compute and upgrade the p values at each sample size($n$). Similarly, we can compute the $LIS$ statistic given the null distribution and alternative distributions are known. If the observations are continuous, we can also compute a simple transformation $S_{ni}= \Phi^{-1}(p_{ni})$, where, $p_{ni}$ is the p value corresponding to the $i$-th coordinate with sample size $n$ and $\Phi()$ is the standard normal CDF. Such transformations are very useful for testing purposes. For continuous random variables, \cite{E04} proposed the use of z-score.  \cite{SC07} showed that, when the both-sided alternatives are asymmetric, the z- score approach has some advantages over the p value based methods. In our analysis, we will take shelter of the z score based method whenever possible. 
			
			\section{Oracle rule using LIS statistics}
			In this section we propose the oracle adaptive rule for general set up. This rule is based on the perspective of a wise oracle who knows the exact joint distribution of the data streams at each time point. Fix $\alpha$, $\beta\in(0,1)$. The rule is described in Algorithm 1.
			\begin{algorithm} 
				\caption{Algorithm for oracle rule in the general set up.}
				
				\begin{algorithmic}[]
					\State Set $n\gets k$ and $\tau\gets0$. \Comment i.e., start with pilot sample size $k$.
					\While {$\tau=0$}
					\State Rearrange  ($t_n^{*1},t_n^{*2},\cdots,$ $t_n^{*m}$) as ($t_{n}^{*(1)} \leq t_{n}^{*(2)} \leq \cdots \leq t_{n}^{*(m)}$).
					\If{$t_n^{*(1)}\leq \alpha$}
					\State $\mathrm{r}_n\gets
					\max\{q \in [m]:\sum_{l=1}^{q}t_n^{*(l)}/q\leq \alpha\} $
					\Else 
					\State $\mathrm{r}_n\gets0$	
					\EndIf
					\If{$t_n^{*(m)}\geq 1-\beta$}
					\State $\mathrm{a}_n \gets
					\max\{q\in[m]:\sum_{l=1}^q(1-t_n^{*(m-l+1)})/q\leq \beta\}$
					\Else 
					\State $\mathrm{a}_n \gets 0$	
					\EndIf
					\If{$\mathrm{r}_n+\mathrm{a}_n<m$}
					\State $n\gets n+1$
					\Else
					\State $\tau\gets n$
					\EndIf
					\EndWhile
					\State For some $s = (m-\mathrm{a}_{\tau}) (1) \mathrm{r}_\tau $ define $\mathfrak{R}^*= \{ i : t_{\tau}^{*i} \leq t_{\tau}^{*(s)} \} $ 
					\State Let,  $\delta_i^*= 1(0)$ if $i \in \mathfrak{R}^* (i\notin \mathfrak{R}^*)$. And let $\boldsymbol{\delta}^* = (\delta_1^*,\delta_2^*\cdots,\delta_m^*)$.
					\State Oracle Rule: $(\tau,\boldsymbol{\delta}^*)$.
				\end{algorithmic}
				
			\end{algorithm}

			\paragraph{Discussion 1:}
			If  $\mathrm{r}_\tau+\mathrm{a}_\tau = m$, we have only one choice of $s$ given by $s=\mathrm{r}_\tau= m - \mathrm{a}_\tau$. i.e., we have a unique decision. But otherwise $s$ can take any one of the values in $\{\mathrm{r}_\tau,\cdots,m-\mathrm{a}_\tau\}$. For a proper choice of $p_1$ we can define $s$ to be \begin{equation*}
				s = \lfloor (m - \mathrm{a}_\tau+1 - \mathrm{r}_\tau ) p_1 \rfloor
			\end{equation*}
			Where $\lfloor.\rfloor$ is the greatest integer function. Now if the probability of a hypothesis to be null is $\pi_0$, the expected number of null hypotheses is $m\pi_0$. So if the probability of a null hypothesis to be rejected is $\alpha_1$ and the probability of an alternative hypothesis to be falsely accepted is $\beta_1$, then the expected total number of rejections is $m(\pi_0\alpha_1+(1-\pi_0)(1-\beta_1))$ and expected total number of acceptance is $m(\pi_0(1-\alpha_1)+(1-\pi_0)\beta_1)$. We equate the mFDR i.e. the proportion of expected false rejections and expected rejection with $\alpha$ and  mFNR i.e. the proportion of expected false acceptance and expected acceptance with $\beta$ to get estimates of $\alpha_1$ and $\beta_1$ as 
			\begin{align*}
				\alpha_1 & = \frac{(1-\pi_0)/ \pi_0 - \beta/(1-\beta)}{(1-\alpha)/ \alpha - \beta/(1-\beta)}\\
				\beta_1 & = \frac{ \pi_0/(1-\pi_0) - \alpha/(1-\alpha)}{(1-\beta)/ \beta - \alpha/(1-\alpha)}
			\end{align*}
			The idea is that those undecided hypotheses potentially have the highest $LIS$ values among those rejected (if they are rejected) or potentially have the lowest $LIS$ values among those accepted. Hence they are most likely to be false positives or false negative decisions. Therefore we want to divide the undecided region into two parts proportional to expected false rejection and expected false acceptance. i.e., we choose $p_1 = \pi_0 \alpha_1 / (1-\pi_0) \beta_1$, where the values of $\alpha_1$ and $\beta_1$ are given as above.
			
			\paragraph{Discussion 2:}
			From algorithm 1, we can define for stage $n$ the  adaptive lower cutoff, denoted by $t_n^{*l}$ as,
			\begin{equation}
				\label{adaporgenlower}
				t_n^{*l}= \begin{cases}
					t_n^{*(\mathrm{r}_n+1)} & \text{ if } \mathrm{r}_n < m \\
					1 & \text{ if } \mathrm{r}_n = m
				\end{cases}
			\end{equation}
			Any hypothesis $i$ with $LIS$ value ($t_n^{*i}$) less than $t_n^{*l}$ is considered to be a potential alternative. Similarly we can define the adaptive upper cutoff for stage $n$ denoted by $t_n^{*u}$ as,
			\begin{equation}
				\label{adaporgenupper}
				t_n^{*u}= \begin{cases}
					t_n^{*(m-\mathrm{a}_n)} & \text{ if } \mathrm{a}_n < m \\
					1 & \text{ if } \mathrm{a}_n = m
				\end{cases}
			\end{equation}
			Any hypothesis $i$ with $LIS$ value ($t_n^{*i}$) greater than $t_n^{*l}$ is considered to be a potential null.
			
			\begin{lemma}
				\label{lma1}
				Suppose all the components of $\mathbf{S_n}$ are continuous. For some $n\in \mathbb{N}$, $\mathrm{r}_n+\mathrm{a}_n\geq m$ if and only if $t_{n}^{*l}>t_{n}^{*u}$ almost surely. Further, if there are some discrete variables present in the mix then, $t_{n}^{*l}>t_{n}^{*u}$ implies $\mathrm{r}_n+\mathrm{a}_n\geq m$. 
			\end{lemma} 
			Proof of lemma \ref{lma1} is in appendix \ref{appendixD}. From lemma \ref{lma1}, we see that, the time taken for arriving at the criteria $t_{n}^{*l}>t_{n}^{*u}$ is equal to $\tau$ almost surely if all the components in $\mathbf{S_n}$ are continuous, whereas, for the presence of discrete random variables, the corresponding time is almost surely greater than or equal to $\tau$. Due to the analytic advantage of this new criteria, we redefine our stopping rule as: 
			
			\begin{equation}
				\label{eqn:newstoppingtime}
				\tau = \inf\{n \in \mathbb{N}: t_{n}^{*l}>t_{n}^{*u}\}
			\end{equation}
			
			To ensure the termination of the procedure at a finite time, we make the following assumption: 
			
			\begin{asmtn}
				\label{asm1}
				Under $H_{0i}$, $t_n^{*i} \stackrel{\text{p}}{\rightarrow}1$ and under  $H_{1i}$, $t_n^{*i} \stackrel{\text{p}}{\rightarrow}0$ $\forall i=1(1)m$.
				
			\end{asmtn}
			This assumption makes sure that as we observe more data, evidence towards the truth increases weakly. 
			
			For some $n\in\mathbb{N}$, $[n]$ denotes the set $\{1,2,\cdots,n\}$
			
			The following theorem proves that the oracle rule for general setup is proper and controls both $FDR$ and $FNR$. 
			
			\begin{theorem}
				\label{thm_1}
				Fix $\alpha,\beta\in (0,1)$. Suppose \textbf{Assumption  \ref{asm1}}  holds and $m\in \mathbb{N}$. Let ($\tau,\boldsymbol{\delta}^*$) be the oracle rule characterized by algorithm 1.  
				then 
				\begin{equation}
					\label{stptm}
					\mathbb{P}(\tau<\infty)=1
				\end{equation} 
				And further, 
				\begin{equation}
					(\tau,\boldsymbol{\delta}^*) \in \Delta(\alpha,\beta)
				\end{equation}
			\end{theorem}
			Proof of Theorem \ref{thm_1} is included in appendix \ref{appendixA}.
			
			In practice, it is hard to compute the $LIS$ values under a general setup. Assumption \ref{asm2} refers to a two-grouped mixture model, where, $LIS$ for each coordinate depends on observation in that coordinate only. Under such premise, $LIS$ are proved to be equal to local FDR ($lfdr$) almost surely. \cite{EfB04} explored the idea of $lfdr$ and its application in multiple testing using a two-component mixture model. In the next section, we assume a similar model and develop an oracle (data-driven) rule using $lfdr$ as the test statistic.
			
			\section{Oracle and data-driven rules under two group mixture model}
			\subsection{Oracle Rule}
			In this section, we consider the following assumptions on the statistics $\mathbf{S}_n$: 
			\begin{asmtn}
				\label{asm2}
				\begin{align}
					\label{eqn:theta.dis}\theta_i \stackrel{iid}{\sim} Benoulli(1-\pi_0)\\
					S_n^i|\theta_i \sim f_n \text{ independently}
				\end{align}
			\end{asmtn}
			
			where, $f_n(.)=\theta_i f_{1n}(.)+ (1-\theta_i)f_{0n}(.)$. The $lfdr$ for $i$ at time $n$ is defined as
			
			\begin{equation}
				\label{lfdr}
				t_n^i=\mathbb{P}(\theta_i=0|Z_n^i)
			\end{equation}
			
			\cite{EfB04} showed that \textbf{Assumption \ref{asm2}} implies $t_n^i=\frac{\pi_0f_{0n}(S_n^i)}{f_n(Z_n^i)}$.
			Now we introduce the following lemma :
			
			\begin{lemma} \label{lma2}
				If \textbf{Assumption \ref{asm2}} is true, then
				\begin{equation*}
					\mathbb{P}(t_n^{*i}= t_n^i)=1 \ \quad \forall i=1(1)m
				\end{equation*}
			\end{lemma}

			This equality occurs as a result of \textbf{Assumption \ref{asm2}}. The proof is in appendix \ref{appendixD}. So, using the $lfdr$ statistics an adaptive oracle rule similar to that of the oracle rule for the general setup may be formulated. Fix $\alpha,\beta\in (0,1)$. The rule is obtained by replacing $t_n^{*i}$ with $t_n^i$ in algorithm 1.

			\paragraph{Discussion:}
			Like before we can define for stage $n$ the  adaptive lower cutoff, denoted by $t_n^{l}$ as,
			\begin{equation}
				\label{adaporindlower}
				t_n^{l}= \begin{cases}
					t_n^{(r_n+1)} & \text{ if } r_n < m \\
					1 & \text{ if } r_n = m
				\end{cases}
			\end{equation}
			Any hypothesis $i$ with $LFDR$ ($t_n^{i}$) less than $t_n^{l}$ is considered to be a potential alternative. Similarly we can define the adaptive upper cutoff for stage $n$ denoted by $t_n^{u}$ as,
			\begin{equation}
				\label{adaporindupper}
				t_n^{u}= \begin{cases}
					t_n^{(m-a_n)} & \text{ if } a_n < m \\
					1 & \text{ if } a_n = m
				\end{cases}
			\end{equation}
			Any hypothesis $i$ with $LFDR$ ($t_n^{i})$ greater than $t_n^{l}$ is considered to be a potential null.
			We can define the stopping time to be 
			\begin{equation}
				T = \inf \{n\in \mathbb{N} : t_n^l > t_n^u\}
			\end{equation}
			This new stopping time, like before, is almost surely equal to that of the original stopping time for the continuous case, but is greater than equal to the original stopping time for the discrete case.
			

			The following theorem shows that the oracle rule under the independent mixture model is proper and controls $FDR$ and $FNR$ at desired levels.
			\begin{theorem}
				\label{thm_2}
				Fix $\alpha,\beta\in (0,1)$. Suppose Assumption \ref{asm1} and \ref{asm2}  holds and $m\in \mathbb{N}$. Let ($T,\boldsymbol{\delta}$) be the oracle rule for the independent mixture model as obtained by replacing $t_n^{*i}$ with $t_n^i$ in algorithm 1. Then, 
				\begin{equation}
					\label{stptm1}
					\mathbb{P}_\theta(T<\infty)=1
				\end{equation} 
				Further, 
				\begin{equation}
					(T,\boldsymbol{\delta}) \in \Delta(\alpha,\beta)
				\end{equation}
			\end{theorem}
			Proof of Theorem \ref{thm_2} is included in appendix \ref{appendixB}.
			
			\subsection{Data-driven Rule:}
			From our earlier discussion, it is established that, $t_n^i= \pi_0f_{0n}(S_n^i)/f(S_n^i)$, but in practice $\pi_0$, $f_{0n}(.)$ and $f_n(.)$ are not known. Therefore, we cannot use the oracle rule directly. In the case of $\mathbf{S_n}$ being a vector of independent continuous random variables, estimates of all of these quantities are readily available in the literature. For example, when the sample size is $n$, if the null distribution of each of the statistics is $F_{0n}$, we can construct the z scores $\mathbf{Z_n}=(Z_n^1,Z_n^2,\cdots,Z_n^m)$ as 
			
			\begin{equation*}
				Z_n^i = \Phi^{-1} (F_{0n}(S_{n^i}))
			\end{equation*}
			
			Then the null distribution for these z scores is standard normal. We can obtain the estimate of the null proportion $\pi_0$ following \cite{CJ10}. The denominator can be obtained using a kernel density estimator. Once all the unknown quantities are estimated, we get the estimated local fdr values as $\hat{t}_n^i, \ i\in \{1,2,\cdots,m\}$ and we can replace $t_n^{*i}$ with $\hat{t}_n^i$ in algorithm 1 to get the data-driven rule for independent mixture model set up for fixed $\alpha,\beta\in(0,1)$.

			\begin{figure}
				\centering
				\includegraphics[scale=0.5]{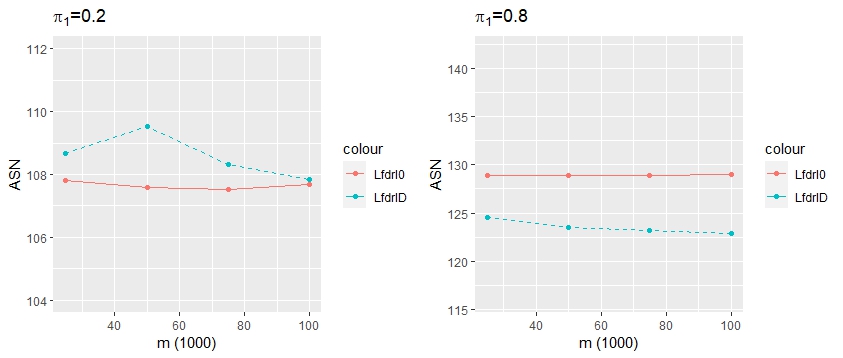}
				\caption{Comparison of average sample size with number of hypotheses }
				\label{fig:1}
			\end{figure}

			\paragraph{Discussion:}
			Like earlier, we can define for stage $n$ the  adaptive lower cutoff, denoted by $\hat{t}_n^{l}$ as,
			\begin{equation}
				\label{adaporindeslower}
				\hat{t}_n^{l}= \begin{cases}
					\hat{t}_n^{(\hat{r}_n+1)} & \text{ if } \hat{r}_n < m \\
					1 & \text{ if } \hat{r}_n = m
				\end{cases}
			\end{equation}
			Any hypothesis $i$ with estimated $LFDR$ ($\hat{t}_n^{i}$) less than $\hat{t}_n^{l}$ is considered to be a potential alternative. Similarly  we can define the adaptive upper cutoff for stage $n$ denoted by $\hat{t}_n^{u}$ as,
			\begin{equation}
				\label{adaporindesupper}
				\hat{t}_n^{u}= \begin{cases}
					\hat{t}_n^{(m-\hat{a}_n)} & \text{ if } \hat{a}_n < m \\
					1 & \text{ if } \hat{a}_n = m
				\end{cases}
			\end{equation}
			Any hypothesis $i$ with estimated $LFDR$ ($\hat{t}_n^{i}$) greater than $\hat{t}_n^{l}$ is considered to be a potential null.
			\begin{lemma}
				\label{lma3}
				Suppose each co-ordinates of $\mathbf{S}_n$ are continuous random variables. For some $n\in \mathbb{N}$, $\hat{r}_n+\hat{a}_n\geq m$ if and only if $\hat{t}_{n}^{l}>\hat{t}_{n}^{u}$ almost surely.
			\end{lemma} 
			Proof of \textbf{Lemma \ref{lma3}} is similar to that of \textbf{Lemma \ref{lma1}} and therefore omitted. Due to \textbf{Lemma \ref{lma3}}, $\hat{t}_n^{u} < \hat{t}_n^{l}$ occurs for the first time at $n=T_d$. So we can restate the data-driven rule by changing the stopping criteria as:
			
			``\textit{Stop sampling as soon as $\hat{t}_n^{l} >\hat{t}_n^{u}$ }"
			
			For the data-driven rule to work properly, we require the estimated local fdrs $\hat{t}_n^i$ to be in close proximity to the actual local fdr values $t_n^i$. We can ensure an asymptotic accuracy of the method i.e. the error rates are asymptotically contained at a prefixed value (Theorem \ref{thm_3}) if the following consistency property for the estimates of the local fdrs holds for each coordinate.
			\begin{asmtn}
				\label{asm3}
				As $m\uparrow \infty$, for all $j \in \{1,2,\cdots,m \}$,
				\begin{equation}
					\hat{t}_n^j\stackrel{p}{\rightarrow} 	t_n^j.
				\end{equation}
			\end{asmtn}
			The following theorem proves that the above procedure controls both $FDR$ and $FNR$ asymptotically. 
			
			\begin{theorem}
				\label{thm_3}
				Fix $\alpha,\beta\in (0,1)$. Define the class of sequential multiple tests that controls $FDR$ and $FNR$ respectively at levels $\alpha$ and $\beta$ asymptotically as:
				\begin{equation*}
					\Delta'(\alpha,\beta) = \{(T,\mathbf{d}) | \lim_{m\uparrow \infty} FDR \leq \alpha \text{ and } \lim_{m\uparrow \infty} FNR \leq \beta \}
				\end{equation*}
				Now suppose assumptions \ref{asm1}, \ref{asm2}, \ref{asm3}  holds and $m\in \mathbb{N}$. Let ($T_d,\hat{\boldsymbol{\delta}}$) be the data-driven rule obtained by replacing $t_n^{*i}$ with $\hat{t}_n^i$ in algorithm 1. Then,  
				\begin{equation}
					(T_d,\hat{\boldsymbol{\delta}}) \in \Delta'(\alpha,\beta)
				\end{equation}
			\end{theorem}
			
			The proof is included in appendix \ref{appendixC}.

			As an intermediate step of proving theorem \ref{thm_3}, we stated and proved lemma \ref{lma12}. This and corollary \ref{cor3} are important results of this paper. Lemma \ref{lma12} says that if the statistics $\mathbf{S}_n$ are continuous and assumptions \ref{asm1},\ref{asm2} and \ref{asm3} are satisfied, then, the stopping time $T_d$ for the data-driven rule converges weakly to a finite non-stochastic positive integer $n_0$ as the number of hypotheses $m$ grows to infinity. corollary \ref{cor3} yields the same result for the oracle stopping time too. This is very important because, in the existing literature, there are no such sequential multiple testing methods that stop in a finite time even when the number of hypotheses gets large indefinitely. As we shall see in the next section, the stopping time of the only competitor of our method, \cite{HB21}, converges to infinity. This result helps our method carry out large-scale multiple testing problems at a very low cost. This phenomenon is illustrated in figure \ref{fig:1}.
			
			\begin{figure}
				\centering
				\includegraphics[scale=0.42]{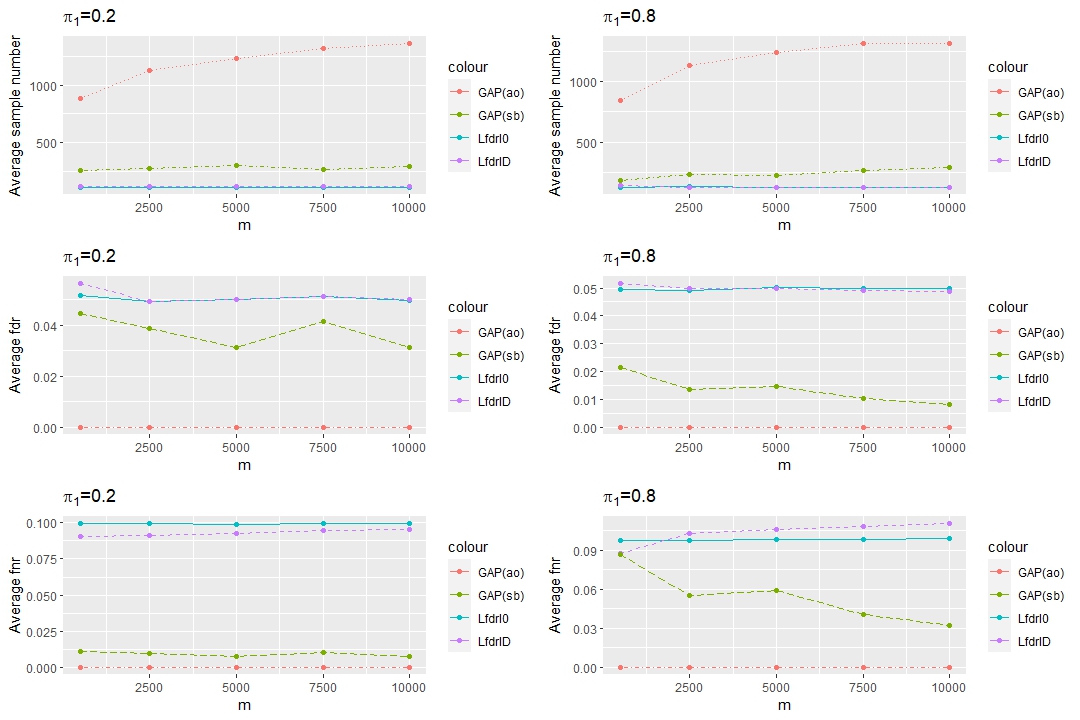}
				\caption{Comparison of average sample size with the number of hypotheses }
				\label{fig:3}
			\end{figure}
			
			\section{Comparison with existing competitor}
			
			\cite{HB21} proved that if the number of alternatives is known ($K$) then among all tests that reject exactly $K$ hypotheses, the Gap Rule defined in \cite{SF17} is asymptotically optimal for simultaneously controlling $FDR$ and $FNR$ as $\alpha$ and $\beta$ tends to $0$, given for each coordinate, a simple null is tested against a simple alternative and observations against each coordinate are iid. They also study the asymptotics of $m$ as functions of $\alpha\wedge\beta$ and provide some restrictions so that the optimality holds. However, the asymptotic behavior of the expected stopping time as $m\uparrow\infty$ has not been studied. In this section we shall show that, for the independence setup described in  \textbf{Assumption} \ref{asm2} stopping time for Gap rule goes to $\infty$ in a weak sense. 
			\begin{lemma}
				\label{lma4}
				Suppose, assumptions  \ref{asm1}, \ref{asm2} hold. Let $m_1$ be the number of alternatives and $T_g^*$ be the stopping time corresponding to the asymptotically optimal Gap rule for controlling $FDR$ and $FNR$ respectively at level $\alpha$ and $\beta$ as mentioned in \textbf{Theorem 3.1} in \cite{HB21}. Then for any $L\in \mathbb{N}$,
				\begin{equation}
					\lim_{m\uparrow \infty} \mathbb{P}(T_g^*<L) = 0
				\end{equation}
			\end{lemma}

            Proof of lemma \ref{lma4} is included in appendix \ref{appendixD}.
			\begin{corollary}\label{cor1}
				If assumptions  \ref{asm1}, \ref{asm2} holds and if $T$ and $T_g^*$ be the stopping times of the oracle rule and asymptotically optimal Gap rule then,
				\begin{equation}
					E(T)=o(E(T_g^*))
			\end{equation} 	\end{corollary}
			
			\begin{proof}
				Proof of \textbf{corollary \ref{cor1}} follows from \textbf{corollary \ref{cor3}} and \textbf{Lemma \ref{lma4}}
			\end{proof}
			
			\section{Simulation Study:}

			
			\begin{table}[H]
				\centering
				\caption{Comparing LfdrI$_0$, LfdrI$_D$ and Gap rule }
				\label{tab:e1comp}
				\resizebox{\columnwidth}{!}{%
					\begin{tabular}{|c|c|c|ccc|ccc|ccc|c|}
						
						\hline
						&\multirow{3}{*}{$m$} & \multirow{3}{*}{$\pi_1$} & \multicolumn{3}{c|}{ $LfdrI_0$} & \multicolumn{3}{c|}{ $LfdrI_D$} & \multicolumn{3}{c|}{$Gapsb$} & $Gapao$\\
						& & & $ASN$ & $\hat{f}dr$(\%) & $\hat{f}nr$(\%)   & $ASN$ & $\hat{f}dr$(\%) & $\hat{f}nr$(\%) & $ASN$ (s\%) & $\hat{f}dr$(\%) & $\hat{f}nr$(\%) & $ASN$ (s\%)  \\
						
						\hline
						\hline
						\multirow{6}{*}{\textbf{E1}} &	\multirow{2}{*}{100}& 0.2 &	109.5 & 4.07 & 10.16 & 107.3 & 5.74 & 10.76 &  227.9 (53) & 4.80 & 1.18 & 563.7 (81) \\
						&& 0.8 &	132.9 & 4.67 & 10.2 & 146.1 & 5.95 & 8.55 &  169.7 (14) & 2.35 & 9.60 & 554.8 (74) \\
						&\multirow{2}{*}{500}& 0.2 &	109.7 & 5.12 & 10.19 & 122.7 & 5.37 & 8.88 &  254.0 (52) & 4.43 & 1.11 & 823.1 (85) \\
						&& 0.8 &	130.6 & 4.89 & 9.83 & 141.5 & 5.42 & 7.96 & 183.8 (23) & 2.15 & 8.64 & 841.4 (83)\\
						&\multirow{2}{*}{1000}& 0.2 &	110.2 & 4.82 & 10.03 & 118.6 & 4.65 & 9.09 & 252.4 (53) & 4.57 &1.13 & 957.4 (88) \\
						&& 0.8 &	130.4 & 4.76 & 10.00 & 138.4 & 5.14 & 8.19 &  176.9 (22) & 2.30& 9.24 & 972.2 (86)\\
						
						\hline
						\hline
						\multirow{6}{*}{\textbf{E2}}& \multirow{2}{*}{100}& 0.2 &	185.04 & 5.03 & 10.14 & 187.55 & 5.89 & 10.77 & 412.2 (55) & 5.06 & 1.23 & 1067.5 (82)\\
						& & 0.8 &	246.9 & 4.39 & 10.06 & 258.24 & 5.39 & 10.94 & 309.6 (17) & 2.49 & 10.83 & 1070.6 (76) \\
						&\multirow{2}{*}{500}& 0.2 &	192.28 & 4.92 & 9.89 & 210.04 & 5.11 & 9.11 &  454.9 (54) & 4.84 & 1.19 & 1584.2 (87) \\
						&& 0.8 &	246.22 & 4.97 & 10.22 & 246.97 & 5.14 & 11.19 &  317.7 (22) & 2.56 & 10.34 & 1594.6 (85)\\
						&\multirow{2}{*}{1000}& 0.2 &	193.13 & 4.86 & 9.93 & 199.3 & 5.03 & 9.73 & 452.9 (56) & 4.89 & 1.23 &  1808.8 (89)\\
						&& 0.8 &	248.5 & 4.98 & 10.19 & 242.6 & 4.88 & 11.9 & 333.4 (27) & 2.36 & 9.44 & 1821.9 (87) \\
						
						\hline
						\hline
						
						\multirow{6}{*}{\textbf{E3}} & \multirow{2}{*}{100} & 0.2 &	90.3 & 4.25 & 10.21 & 101.6 & 5.63 & 9.62 & 192.9 (47) & 4.28 & 1.03 & 296.2 (66)\\
						& & 0.8 &	124.6 & 4.61 & 11.0 & 122.8 & 4.86 & 14.25 &155.1 (21) & 2.39 & 9.66 & 297.8 (59)\\
						& \multirow{2}{*}{500} & 0.2 &	91.5 & 5.22 & 9.88 & 93.1 & 5.61 & 9.87 &  229.9 (60) & 4.97 & 1.26 & 504.5 (82)\\
						& & 0.8 &	128.1 & 4.93 & 9.87 & 120.5 & 4.63 & 13.99 & 165.4 (27) & 2.37 & 9.85 & 507.9 (76)\\
						&\multirow{2}{*}{1000}& 0.2 &	91.3 & 5.21 & 9.89 & 93.1 & 5.14 & 9.94 &  231.16 (60) & 4.89 & 1.20 &  590.7 (84) \\
						& & 0.8 &	128.0 & 4.95 & 9.90 & 118.9 & 4.63 & 13.96 &  168.4 (30) & 2.31 & 9.54 & 612.2 (81)\\
						
						\hline
						\hline
						
						\multirow{6}{*}{\textbf{E4}}&\multirow{2}{*}{100}& 0.2 &	41.7 & 4.30 & 9.63 & 36.7 & 6.05 & 11.02 & 75.88 (52) & 4.09 & 0.96 & 170.7 (79)\\
						& & 0.8 &	53.3 & 4.79 & 8.92 & 53.4 & 4.75 & 9.50 & 65.13 (18) & 1.77 & 7.07 & 170.1 (69)\\
						&\multirow{2}{*}{500}& 0.2 &	41.5 & 4.91 & 9.82 & 41.4 & 5.55 & 9.00 & 104.9 (61) & 3.15 & 0.8 & 275.2 (85)\\
						& & 0.8 &	53.1 & 5.01 & 8.68 & 53.1 & 5.22 & 8.31 & 103.68 (49) & 0.84 & 3.44 & 278.7  (63)\\
						&\multirow{2}{*}{1000} & 0.2 &	41.6 & 4.96 & 9.32 & 43.0 & 5.82 & 8.36 & 134.21 (68) & 1.64 & 0.4 & 320.7 (87)\\
						& & 0.8 &	53.2 & 5.04 & 9.07 & 53.3 & 5.21 & 8.42 & 133.75 (60)  & 0.4 & 1.63 & 324.8 (84)\\
						\hline
					\end{tabular}
				}
			\end{table}

			\begin{table}[H]
				\centering
				\caption{Comparing LfdrI$_D$ with BH rule }
				\label{tab:e2comp}
				\resizebox{\columnwidth}{!}{%
					\begin{tabular}{|c|c|c|ccc|ccc|}
						\hline
						&\multirow{2}{*}{$m$} & \multirow{2}{*}{$\pi_1$} & \multicolumn{3}{c|}{ LfdrI$_D$} & \multicolumn{3}{c|}{BH} \\
						&&& $ASN$ & $\hat{f}dr$ (\%) & $\hat{f}nr$ (\%) & $\hat{n}$ (s(\%)) & $\hat{f}dr$ (\%) & $\hat{f}nr$ (\%) \\
						
						\hline
						\hline
						\multirow{6}{*}{\textbf{Ex 5}}&\multirow{2}{*}{2500}& 0.2 & 267.76 & 4.97 & 12.7 & 332.0 (19) & 4.46  & 12.2  \\
						&& 0.8 & 476.92 & 5.23 & 9.82 & 962.0 (50) & 0.98 & 9.8  \\
						&\multirow{2}{*}{7500} & 0.2 & 328.92 & 4.51 & 10.70 & 414.0 (21) & 4.07  & 10.51   \\ 
						&& 0.8 & 469.52 & 	5.16 & 10.49 & 954.0 (51) & 0.98 & 10.51 \\
						&\multirow{2}{*}{10000}& 0.2 & 357.48 & 5.05 & 9.70 & 428.0 (17) & 3.93   & 9.79 \\
						&& 0.8 & 467.36 & 5.06 & 10.6 & 944.0 (51) & 0.98 & 10.51 \\
						\hline
						\hline
						\multirow{6}{*}{\textbf{Ex 6}}&\multirow{2}{*}{2500} & 0.2  & 121.72 & 5.03 & 8.8 & 151 (20) & 4.28 & 8.81  \\
						&& 0.8 &	 133.12 & 5.12 & 9.19 & 261 (49) & 0.99 & 9.37  \\
						&\multirow{2}{*}{7500}& 0.2  & 115.28 & 4.81 & 9.48 &  143 (19) & 4.07 & 9.51 \\
						&& 0.8 & 130.12 & 5.02 & 9.84 & 258 (50) & 1.00 & 9.88 \\
						&\multirow{2}{*}{10000} & 0.2  & 113.7 & 4.94 & 9.65 &  142 (20) & 3.98 & 9.66 \\
						&& 0.8 & 129.32& 5.00 & 10.04 &  257 (50) & 1.04 & 9.96 \\

						\hline
						\hline
						
						\multirow{6}{*}{\textbf{Ex 7}}&\multirow{2}{*}{2500} & 0.2  & 209.78 & 5.78 & 9.99 & 234 (10) & 3.91 & 9.86 \\
						
						&& 0.8  & 261.22 & 4.96 & 10.33 & 443 (41)& 0.91 & 10.26  \\
						
						&\multirow{2}{*}{7500} & 0.2  & 206.64 & 5.66 & 10.12 &  231 (11) & 3.65 & 10.29  \\
						
						&& 0.8  & 260.08 & 4.91 & 10.81 &  435 (40) & 0.90 & 10.97 \\
						
						&\multirow{2}{*}{10000} & 0.2  & 207.56 & 5.73 & 9.92 & 234 (11) & 3.82 & 9.91  \\
						
						&& 0.8  & 275.06 & 5.14 & 8.38 & 460 (40) & 1.00 & 8.35  \\

						\hline

					\end{tabular}
				}
			\end{table}
			
			In this section, our goal is to numerically validate the theoretical claims made earlier. The foremost of them was that the performance of the data-driven rule is asymptotically equivalent to that of the oracle rule. We also ensured theoretically that the expected stopping time of the Gap rule as proposed by \cite{HB21} is much higher than that of the oracle rule. Numerical validation for the same is provided here. Also, the adaptive Gap rule, where the cutoff boundary for making decisions is obtained using monte carlo simulations, can not be compared theoretically and therefore the comparison is made using numerical examples. Finally, we measured the performance of our sequential multiple testing rule with the Benjamini-Hochberg rule. Each numerical example is considered for two cases, one sparse case, where $\pi_1$ is 0.2 and one dense case with $\pi_1$ equal to 0.8.
			
			We assume that, $\boldsymbol{\theta}_0 = (\theta_{01},\theta_{02},\cdots,\theta_{0m})\in \{0,1\}^m$ is a realised value on the unobservable random binary indices $\boldsymbol{\theta}$ according to (\ref{eqn:theta.dis}). Given $\boldsymbol{\theta}_0$, the data $\{X_n^i| n \in \mathbb{N} , \ i \in [m]\}$ is generated from the mixture model 
			\begin{equation}
				\label{eqn:datageneration}
				X_n^i \stackrel{iid}{\sim} (1-\theta_{0i}) f_0(x) + \theta_{0i} f_1(x)
			\end{equation}
			
			For the oracle rule the computation of local fdr is straightforward since the parameters and the distributions are assumed to be known. For the data-driven rule, null proportion $1-\pi_1$ is estimated using the method described in \cite{CJ10} for the continuous case. Here, we have computed the z scores since the method requires the null distribution to be Gaussian. For the continuous case, the mixture density is estimated using the kernel density estimation method. The null distribution is assumed to be standard normal. The estimation method by \cite{CJ10} requires a particular choice of the tuning parameter $\lambda$ which was fixed at 0.1. The only case when a discrete distribution was assumed was to test whether the observation follow Bernoulli(7,0.1) ($f_0$) or a Bernoulli(7,0.3)($f_1$) distribution. In that case, the null proportion $1-\pi_1$ was estimated using the method described by \cite{STS04}. With that established, the other parts in the expression of local fdr are known since the pairs of hypotheses in question are both simple. Once we got the local fdr values, at each sample size, we follow the steps described in Algorithm 1. The corresponding code is very easy to implement and will be uploaded soon. Here, throughout this section, we have fixed $\alpha$ at 5\% and $\beta$ at 10\%.
			
			\paragraph{Numerical study 1:} In this example, we compare the performance of the oracle rule ($LfdrI_0$) with the data-driven rule ($LfdrI_D$). As a data generating process, we consider the \textbf{example 1} with $\mu_0 = 0 $ and $\mu_1= 0.25$. The average sample numbers ($ASN$) were computed based on 50 Monte Carlo runs. Figure \ref{fig:1} shows an $ASN$ versus $m$ plot. The plots for both sparse and dense cases support our theoretical claim.
			
			\paragraph{Numerical study 2:} The performance of $LfdrI_0$ and $LfdrI_D$ are measured against that of the GAP rules as described in \cite{HB21}. To implement the GAP rule, we need to be informed about the number of rejections. To make it comparable with our method, we generate the data using (\ref{eqn:datageneration}), we note the number of rejections for each case and use it for implementing the GAP rule. We also note that the implementation of GAP rule really works if the hypotheses to be tested are simple. We carefully choose such cases. The authors in \cite{HB21} provide one asymptotically optimal rule ($GAPao$), where a theoretical value of the cutoff is provided. For practical purposes, another adaptive rule is also mentioned, where the boundary is obtained by simulations. Here, whenever this rule ($GAPsb$) has been implemented, we have used 50 Monte Carlo runs and computed the average false discovery proportions ($\hat{f}dr$) and the average false nondiscovery proportions ($\hat{f}nr$). We have repeated this process for different values of the cutoff starting from 0.1 with a gap of 0.1 until we get a cutoff $\hat{c}$ for which $\hat{f}dr\leq \alpha$ and $\hat{f}nr\leq \beta$. Then for $\hat{c}$, we again compute the $ASN$, $\hat{f}dr$ and $\hat{f}nr$ with 200 Monte Carlo runs.
			
			
			We consider \textbf{example 1} with $\mu_0 = 0$ and $\mu_1=0.25$. Then for $LfdrI_0$, $LfdrI_D$, $GAPao$ and $GAPsb$, we generate plots of $ASN$, $\hat{f}dr$ and $\hat{f}nr$ (each computed for 200 Monte Carlo runs) against $m$ in figure \ref{fig:3}.
			
			
			In table \ref{tab:e1comp}, we list the $ASN$, $\hat{f}dr$ and $\hat{f}nr$ of $LfdrI_0$, $LfdrI_D$, $GAPao$ and $GAPsb$ in different testing problems as will be discussed latter. In each case, we used 200 Monte Carlo runs to compute different components. We write down the savings in $ASN$ for using $LfdrI_D$ as compared to $GAPao$ and $GAPsb$. For the $GAPao$ rule, the $\hat{f}dr$ and $\hat{f}nr$ has been omitted owing to the fact that, in each case, both of them have an average value of 0. 
			
			The exaples considered in table \ref{tab:e1comp} are listed below.
			
			\textbf{Example 1:} Here, $X_n^i|(\theta_i=\theta_{0i}) \stackrel{iid}{\sim} (1-\theta_{0i}) N(\mu_0,1) + \theta_{0i} N(\mu_1,1)$. For table \ref{tab:e1comp}, we consider $\mu_0 =0$ and $\mu_1 = 0.25$.
			
			\textbf{Example 2:} Here, $X_n^i|(\theta_i=\theta_{0i}) \stackrel{iid}{\sim} (1-\theta_{0i}) exp (\lambda_0)+ \theta_{0i} exp (\lambda_1)$. For table \ref{tab:e1comp}, we consider $\lambda_0 =1$ and $\lambda_1 = 1.2$.
			
			\textbf{Example 3:} Here, $X_n^i|(\theta_i=\theta_{0i}) \stackrel{iid}{\sim} (1-\theta_{0i}) N(0,\sigma_0) + \theta_{0i} N(0,\sigma_1)$. For table \ref{tab:e1comp}, we consider $\sigma_0 =1$ and $\sigma_1 = 1.2$.
			
			\textbf{Example 4:} Here, $X_n^i|(\theta_i=\theta_{0i}) \stackrel{iid}{\sim} (1-\theta_{0i}) Bernoulli(7,p_0) + \theta_{0i} Bernoulli(7,p_1)$. For table \ref{tab:e1comp}, we consider $p_0 =0.1$ and $p_1 = 0.3$.
			
			For each of these examples, we observe that, the performance of $LfdrI_0$ and $LfdrI_D$ are equivalent. The most interesting observation here is that, on both of the cases, $\hat{f}dr$ and $\hat{f}nr$ are almost equal to desired levels $\alpha$ and $\beta$ respectively. This proximity is more prominant and stable in case of $LfdrI_0$ in comparison with $LfdrI_D$ which supports the etsablished theory that, $FDR$ and $FNR$ control is exact in case of the oracle rule and is assymptotic in case of the data-driven rule. We can therefore hope to prove our method to be optimal in some sense in a future work. In each of the cases, $GAPao$ results in 0 $\hat{f}dr$ and $\hat{f}nr$ with 0 standard error and therefore omitted in the table. Savings with respect to $LfdrI_D$ is reported for $GAPao$ and $GAPsb$. This savings is much higher in case of $GAPao$ due to its conservative property. It is somewhat less but still significant in case of $GAPsb$. Therefore table \ref{tab:e1comp} shows absolute dominance of $LfdrI_D$ over the $GAP$ rule. Here the difference in savings for the sparse case and the dense case results from the asymmetry in the values of $\alpha$ and $\beta$ chosen here. Similar resrults are obtained in case of large values of $m$ with respect to the example 1 as shown in figures \ref{fig:1} and \ref{fig:3}.
			
			\paragraph{Numerical study 3:} This study is so designed that, we can compare the performance of the sequential multiple testing method $LfdrI_D$ with some existing fixed sample rules- among which, that (BH) devised by \cite{BH95} is the most popular. Also, being an $FDR$ controlling rule, this is a matching competetor for us. For comparison purpose, we consider different hypotheses testing problems. In each cases, first, we run $LfdrI_D$ for 200 Monte Carlo runs with $\alpha$ and $\beta$ fixed at level mentioned earlier and note the $ASN$, $\hat{f}dr$ and $\hat{f}nr$. Then, we continue run BH with same $\alpha$ but different sample size $n$ and note $\hat{f}nr$ untill we get a $\hat{n}$ that yields a $\hat{f}nr$ which is very close to that obtained from $LfdrI_D$. We report $\hat{n}$, $\hat{f}dr$ and $\hat{f}nr$. We also report savings in the $ASN$ of $LfdrI_D$ as compared with $\hat{n}$  of BH. We use the BH function in the library "Mutoss" in the Bioconductor repository for this purpose. The results are reported in table \ref{tab:e2comp}. Here the following examples have been considered.
			
			\textbf{Example 5:} At each stage $n\in\mathbb{N}$,  for each coordinate $i\in[m]$, 2 observation are made. $X_n^i$ is invariably coming from the null distribution (control) and $Y_n^i$ is coming from a mixture density (case). i.e., this is a sequential case-control or 2 sample study. For single hypotheses framework, such problem has a wide literature. One can refer to \cite{CG84} for some insights. Define $n_X(n_Y)$ to be the number of new observations of the control (case) variable to be made at each stage. Here we assumed $n_X=n_Y=1$. The alternative density ($f_1(x)$) of the mixture distribution is considered to be bimodal. i.e., 
			\begin{equation*}
				\begin{split}
					X_n^i &\stackrel{iid}{\sim} N(\mu_0,\sigma_0)\\
					Y_n^i | (\theta_i=\theta_{0i}) \stackrel{iid}{\sim} (1-\theta_{0i}) N(\mu_0,\sigma_1) &+ \theta_{0i} (\eta_i N(\mu_1,\sigma_1) + (1-\eta_i) N(\mu_2,\sigma_1))
				\end{split}
			\end{equation*}
			with $\eta_i\stackrel{iid}{\sim} Bernoulli (p_1)$. Here, we considered $\mu_0 = 0$, $\mu_1=0.25$, $\mu_2 = -0.5$, $\sigma_0=\sigma_0=1$ and $p_1=0.75$. We performed a 2 sample Welch test for each of the coordidnates. This Welch statistics becomes our intended statistics and we preform $LfdrI_D$ based on these test statistics. We use the p values from these tests to perform BH.

			\textbf{Example 6:} We assume 
			\begin{equation*}
				X_n^i|(\theta_i=\theta_{0i}) \stackrel{iid}{\sim} (1-\theta_{0i})N(\mu_0,\sigma^2)+\theta_{0i} N(\mu_1,\sigma^2)
			\end{equation*}
			
			Here, we considered $\mu_0 = 0$, $\mu_1 = 0.25$, $\sigma=1$. We performed the Student's t test for this example. the $LfdrI_D$ and BH was performed as before based on the test statistics.
			
			\textbf{Example 7:} We assume 
			\begin{equation*}
				X_n^i|(\theta_i=\theta_{0i}) \stackrel{iid}{\sim} (1-\theta_{0i})\text{Cauchy}(\mu_0,1)+\theta_{0i} \text{Cauchy}(\mu_1,1)
			\end{equation*}
			
			Here, we considered $\mu_0 = 0$, $\mu_1 = 0.25$. As the test statistics, we calculated the log likelihood ratio. The standard distribution of the statistic is not known. Therefore, at each stage, we simulated 100000 null values of the statistics and we estimated the null CDF values using these simulated values and we obtained the z score and we proceeded as before. For computing the p values, we used the same simulated null log-likelihood ratios and the fact that a high likelihood ratio provides evidence against the null hypotheses. i.e., one-sided test was appropriate.
			
			For each of these examples, we observe that the savings for using $LfdrI_D$ with respect to $BH$ are large. The discrepancy in savings in sparse and the dense case is due to the difference in the proportion of null ($1-\pi_1$), since it plays a major role in control of $FDR$ for $BH$ method.
			
			\section{Data Application}
			
			In this section, we apply our method in real datasets and compared them against the Benjamini Hochberg method and the local fdr based method proposed by \cite{SC07}. We have used 2 datasets for this purpose. The first of them is the Prostate Data (\cite{Singhea02}) obtained from  \href{https://efron.ckirby.su.domains/LSI/datasets-and-programs/}{https://efron.ckirby.su.domains/LSI/datasets-and-programs/}. The dataset contains gene expression data on 6033 genes from 102 prostate cells. Among them, 52 were from tumor cells and the rest were from nontumor cells. Therefore, this constitutes a case-control study and the goal is to identify genes whose expressions are associated with prostate tumors.
			
			Suppose, $X_{ij}$ represents the $j$ th gene expression associated with the $i$ th prostate cell with tumor and $Y_{kj}$ represents the $j$ th gene expression associated with the $k$ th prostate cell without tumor.  Here , $j = 1(1)6033$, $i = 1(1)52$ and $k = 1(1)50$. We have performed 2 sample t test against each of the genes. For the $i$-th gene we have computed the statistic 
			
			\begin{equation*}
				S_{n}^i = \frac{\bar{X}_{n1}^i - \bar{Y}_{n2}^i}{\hat{\sigma}_{n1,n2}^i\sqrt{\frac{1}{n1}+\frac{1}{n2}}}
			\end{equation*}
			
			Where, $n1$ is the number of observations from the prostate tumor cells, $n2$ is the number of samples from the prostate cells without tumor, and $n=n1+n2$. $\bar{X}_{n1}^i$, $\bar{Y}_{n2}^i$ and $\hat{\sigma}_{n1,n2}^i$ are respectively the mean of the $i$ the gene expression for $n1$ tumor cells, that of the $n2$ non-tumor cells and the pulled standard deviation (with denominator $n-2$) fot the full data with $n1$ and $n2$ samples for the $i$ th gene. This is the two-sample t statistic corresponding to the sample of sizes $n1$ and $n2$. Under the assumption that 
			\begin{equation*}
				X_{ij} \stackrel{iid}{\sim} N(\mu_1,\sigma^2) 
			\end{equation*}
			and 
			\begin{equation*}
				Y_{ij} \stackrel{iid}{\sim} N(\mu_2,\sigma^2) 
			\end{equation*}
			for all $i=1(1)6033$, $S_n^i$ follows t distributions with $n-2$ degrees of freedom independently. 
			
			After standardizing and quantile normalizing the dataset, we compute for the full data, i.e., for $n1=52$ and $n2=50$, the statistic $S_{102}^i$ corresponding to the $i$-th gene,  $i=1(1)6033$. Our objective is to test 
			
			\begin{equation*}
				H_{0i}: \mu_1 = \mu_2 \quad \text{versus} \quad  H_{1i}: \mu_1 \neq \mu_2
			\end{equation*}
			
			Therefore,  a both-side test is appropriate here. We compute the p values accordingly. i.e.,
			
			\begin{equation*}
				p_n^i = 2 \min \{F_{100}(S_n^i),(1-F_{100}(S_n^i)) \}
			\end{equation*}
			
			Here, $F_n$ corresponds to the cdf of a t distribution with $n$ degrees of freedom.
			
			Similarly, we calculate the z scores as 
			\begin{equation*}
				z_n^i =  \Phi^{-1}(F_{100}(S_n^i))
			\end{equation*}
			
			The following picture shows the histogram of the p values and z scores for the full data.
			
			\begin{figure}[H]
				\label{fig:4}
				\centering
				\includegraphics[scale=0.48]{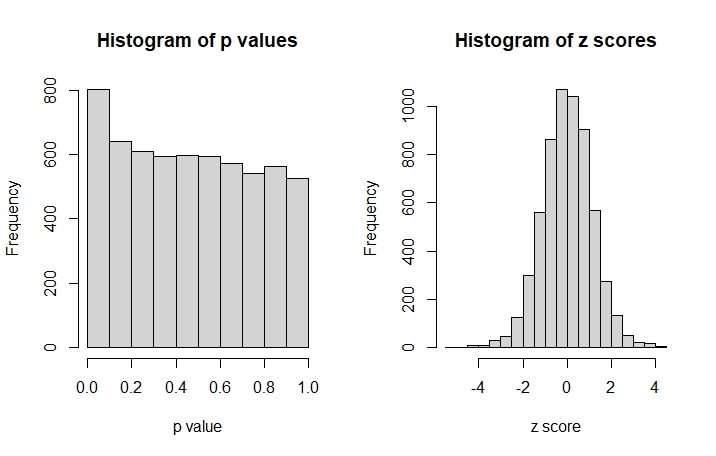}
				\caption{Histogram of P values and Z scores for the full Prostate dataset.}
			\end{figure}
			
			The histograms show that the z scores follow a normal distribution and the p values follow a uniform distribution (except perhaps for a larger frequency near 0), which justifies our assumption of normality for the data. The graphs indicate that the proportion of alternatives is relatively very small (i.e., this data is sparse).
			
			First, we apply the Benjamini Hochberg rule (\cite{BH95}) for the full data. For this, we set $\alpha=0.05$. The number of positive genes we thus obtained is 21. We follow this by the AdaptZ method described by \cite{SC07} with the same value of $\alpha$. Here we obtained a total of 29 genes responsible for prostate cancer.
			
			Finally, we apply our sequential multiple testing method to the dataset. We start with a pilot sample of 50 cells, among which, 25 were from the tumor cells (case) and 25 were from the normal cells (control). We then computed the z scores corresponding to the two-sample t-test statistics (under null which follow independent t distributions with parameter 48) against each gene expression value from the 50 cases and 50 controls. These z scores are used to compute the local fdr estimates for each of the 6033 genes. We then follow the second step of the algorithm of $LfdiI_D$ to obtain the number of potential nulls ($\hat{a}_{50}$) and potential alternatives ($\hat{r}_{50}$) at levels $\alpha=0.05$ and $\beta=0.1$. We note that, here  $\hat{r}_{50} = 14$ and $\hat{a}_{50}=5982$. i.e., $\hat{r}_{50}+\hat{a}_{50} = 5996$ which is less than 6033. So we proceed to the next stage. For each consecutive step, we obtain a new sample from either the case or the control with an equal probability. With this updated sample of size 51, we repeat the same process. We continue this until we have for some $n$, $\hat{r}_{n}+\hat{a}_{n} \geq 6033$. This $n$ is called the stopping time.
			
			For $\beta =0.1$, the stopping time is 69 and we can discover 12 genes to be positive. Here, the savings of the sample size is high (almost 32.4 \%) But the number of discoveries is much less than both Benjamini Hochberg and AdaptZ methods. If we consider $\beta=0.07$ however, The stopping time is 87 (with a sample size savings of about 15 \%), but our method discovers 23 positive genes. Which is higher than the number of discoveries made by the BH method.
			
			The second dataset we have considered, is the dataset collected and used by \cite{Gea99}. The dataset consists of gene expression levels on 7129 genes from the bone marrow tissues of 72 acute leukemia patients. Among them, 47 were suffering from acute lymphoblastic leukemia,
			(ALL) and the rest of the 25 were acute myeloid
			leukemia (AML) patients. Our goal is to discover genes that discriminate between these two types of leukemia.
			
			For this, we consider a mixture of gaussian distribution as before for each gene expression value for each patient. Therefore, a two-sample t-test is applicable here as well. The following picture shows the histograms of the z scores and p values for the full dataset.
			
			\begin{figure}[H]
				\centering
				\label{fig:5}
				\includegraphics[scale=0.7]{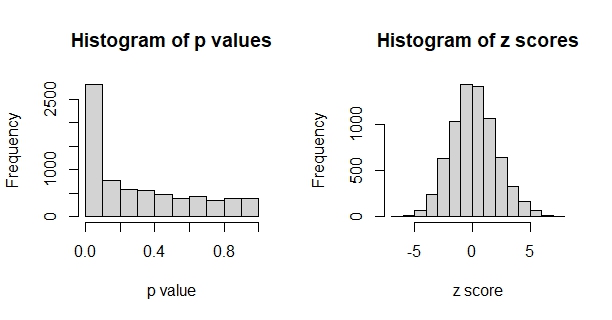}
				\caption{Histograms of p value and z scores of Golub data}
			\end{figure}
			
			The histogram of p value has a peak near 0, otherwise, it seems to be uniformly distributed. The Histogram of the z scores resembles a normal distribution with slightly heavy frequencies at the tail. This tells us that the distribution is medium sparse.
			
			The BH method for $\alpha$ 0.05 discovers 1280 positive genes. AdaptZ method for the same level of $\alpha$ can identify 1300 genes that can identify the difference between these two types of leukemia.
			
			The $LfdrI_D$ method for $\alpha=0.05$ and $\beta=0.1$ does not stop for the given dataset and it also identifies 1300 positive genes.
			
			So, from the above data applications, we can conclude that, the data-driven method can either save sample size or performs as well as the optimal AdaptZ rule for the nonsequential case. 
			
			\appendix
			\section{Proof of Theorem 1} \label{appendixA}
			\begin{proof}
				Let $\mathcal{A}=\{
				i \in [m] : \theta^i=1
				\}$ be the set of signals. Note that, if at stage $n$, $t_n^{*i}\leq\alpha \ \forall i \in \mathcal{A}$ and $t_n^{*i}\geq(1-\beta)\ \forall i \in \mathcal{A}^c$, then we must have, $\tau\leq n$. i.e.,
				
				\begin{equation*}
					\begin{split}\mathbb{P}(\tau\leq n) \geq  & \mathbb{P}(\{\cap_{i \in \mathcal{A}}\{ t_n^{*i}\leq\alpha \}\} \cap \{\cap_{i \in \mathcal{A}^c}\{ t_n^{*i}\geq(1-\beta) \}\} )\\
						\geq & \sum_{i \in \mathcal{A}} \mathbb{P}^i_1(t_n^{*i}\leq\alpha) + \sum_{i \in \mathcal{A}^c}\mathbb{P}^i_0 (t_n^{*i}\geq(1-\beta))-m+1 \\
						\geq & 1- \sum_{i \in \mathcal{A}} \mathbb{P}^i_1(t_n^{*i}>\alpha) - \sum_{i \in \mathcal{A}^c}\mathbb{P}^i_0 (t_n^{*i}<(1-\beta)) \\
					\end{split}
				\end{equation*}
				The second line is due to Boole's inequality. Here, $\mathbb{P}^i_j(A)=\mathbb{P}(A|\theta_i = j)$ $\forall i = 1(1)m$ and $j=0,1$. The proof is complete if we take limit $n\uparrow\infty$ to both side of the inequality and use assumption \ref{asm1}.
				
				To prove the second part, first note that for timepoint $n$, both $ \sum_lt_n^{*(l)}/q$ and $1/q\sum_{l=1}^q(1-t_n^{*(m-l+1)})$ are increasing in $q\in[m]$ almost surely. Now we introduce the following 2 lemmas:
				\begin{lemma}
					\label{lma5}
					If $l \in [\mathrm{r}_{\tau}]$ almost surely, the test $(\tau,\mathbf{D}')$ with $\mathbf{D}'=(D'^1,D'^2,\cdots$, $D'^m)$ given by
					\begin{equation*}
						D'^i = \mathbbm{1} (t_{\tau}^{*i} \leq t_{\tau}^{*(l)}) 
					\end{equation*}
					controls $FDR$ at level $\alpha$ for all $\theta\in\{0,1\}^m$.
				\end{lemma}
				
				and 
				\begin{lemma}
					\label{lma6}
					If $l \in [\mathrm{a}_{\tau}]$ almost surely, the test $(\tau,\mathbf{d}')$ with $\mathbf{d}'=(d'^1,d'^2,\cdots$, $d'^m)$ given by
					\begin{equation*}
						d'^i = \mathbbm{1} (t_{\tau}^{*i} \geq t_{\tau}^{*(m-s+1)})
					\end{equation*}
					controls $FNR$ at level $\beta$ for all $\theta\in\{0,1\}^m$.
				\end{lemma}
				
				Proofs of \textbf{Lemma \ref{lma5}} and \textbf{\ref{lma6}} are provided in appendix \ref{appendixD}.
				
				Now, for any $s = (m-\mathrm{a}_\tau)(1)\mathrm{r}_\tau $, and rejection region $\mathfrak{R}^*$ given by algorithm 1, $FDR\leq\alpha$ by \textbf{Lemma \ref{lma5}} and for acceptance region $[m]\setminus \mathfrak{R}^*$, $FNR\leq\beta$ by \textbf{Lemma \ref{lma6}}
				
				So, $(\tau,\boldsymbol{\delta}^{*})\in \Delta(\alpha,\beta)$  using \textbf{Lemma \ref{lma5}} and \textbf{\ref{lma6}}.
			\end{proof}
			
			\section{Proof of Theorem 2} \label{appendixB}
			\begin{proof}
				First we state the following lemma:
				\begin{lemma}
					\label{lma7}
					Provided assumption \ref{asm2} is true, for any random finite stopping time $T$, defined on $\{\sigma_n\}$,
					\begin{equation}
						\mathbb{P}_\theta(t_T^{*i}\neq t_T^{i})=0
					\end{equation}
				\end{lemma}
				Proof of \textbf{Lemma \ref{lma7}.} is in appendix \ref{appendixD}. \textbf{Lemma \ref{lma7}.} ensures that, although $T$ depends on all the $m$ data-streams, due to assumption \ref{asm2}, for making inference about the $i$-th hypothesis, observations from the $i$-th data-stream are sufficient. The proof of theorem \ref{thm_2} follows from theorem \ref{thm_1} and \textbf{Lemma \ref{lma7}.}.
			\end{proof}
			
			\section{Proof of Theorem 3} \label{appendixC}
			
			\begin{proof}
				
				As discussed earlier, the hypothesis corresponding to data-stream $i$ with $\hat{t}_n^i<\hat{t}_n^l$ are considered potential alternatives and the hypothesis corresponding to data-stream $i$ with $\hat{t}_n^i>\hat{t}_n^u$ are considered potential nulls. If there are some data-streams $i$ with $\hat{t}_n^l \leq \hat{t}_n^i \leq \hat{t}_n^u $, we observe a new sampling unit. It is evident that if for some $n$, $\hat{t}_n^l>\hat{t}_n^u$, each data-stream is either a potential null or a potential alternative (or both!). So we can stop our sampling procedure. \textbf{Lemma \ref{lma4}} says that for continuous test statistics, the opposite is also true; i.e., at stopping time $T_d$, we must have $\hat{t}_{T_d}^l>\hat{t}_{T_d}^u$ with probability 1.

				From the discussion above, we get an alternative definition of the stopping time $T_d$; namely, we stop at the earliest time $n$ when the adaptive rejection boundary $t_{n}^l$ is greater than the adaptive acceptance boundary $t_n^u$. i.e., $T_d=\inf\{n\in \mathbb{N}:\hat{t}_{n}^l>\hat{t}_{n}^u\}$. This definition widens the scope to study the asymptotic behavior of the stopping time as the number of data-streams $m$ diverges to $\infty$. But first, we need to establish the asymptotic properties of $\hat{t}_{n}^l$ and $\hat{t}_{n}^u$. Define,
				
				\begin{equation}
					\label{dfn8}
					\hat{Q}_n(t) = \begin{cases}
						\frac{\sum_{j=1}^m\mathbbm{1}(\hat{t}^j_n\leq t)\hat{t}^j_n}{\sum_{j=1}^m\mathbbm{1}(\hat{t}^j_n\leq t)} & t \in [t_n^{(1)},1]\\
						0 & t\in[0,t_n^{(1)})
					\end{cases}
				\end{equation}
				
				The following lemma describes some properties of the function $\hat{Q}_n(t)$:
				\begin{lemma}
					\label{lma8}
					\begin{enumerate}
						\item 	$\hat{Q}_n(t)$ is constant in the interval $[\hat{t}_n^{(r)},\hat{t}_n^{(r+1)})$ for $r=1(1)$ $m-1$ and in the intervals $[0,\hat{t}_n^{(1)})$ and $[\hat{t}_n^{(1)},1]$
						\item   $\hat{Q}_n(t)$ is right-continuous in $(0,1)$.
						\item $\hat{Q}_n(t)$ is non-decreasing in $[0,1]$.
					\end{enumerate}	
				\end{lemma}
				Proof of \textbf{Lemma \ref{lma8}} is in appendix \ref{appendixD}. Now note that, if $\hat{r}_{n}=m$, by definition of $\hat{r}_{n}$,  $\hat{Q}_n(1)\leq\alpha$. So, $  \sup\{t\in [0,1]: \hat{Q}_n(t)\leq\alpha\}=1 $ and if $\hat{r}_{n}<m$, $ \sup\{t\in [0,1]: \hat{Q}_n(t)\leq\alpha\}=\hat{t}_n^{(\hat{r}_{n}+1)} $. So,
				\begin{equation}
					\label{dfn9}
					\sup\{t\in [0,1]: \hat{Q}_n(t)\leq\alpha\}=\hat{t}_{n}^l
				\end{equation}
				
				The following lemma ensures a weak non-stochastic limit to $ \hat{t}_{n}^l $.
				
				\begin{lemma}
					\label{lma9}
					Define $\mathcal{Q}_n(t)=\frac{\pi_0 \mathbb{P}^1_0(t_n^1\leq t)}{\mathbb{P}_\theta(t_n^1\leq t)\vee1}$ for $t\in[0,1]$.Let,
					$\mathcal{T}_n^l=\sup\{t \in [0,1]: \mathcal{Q}_n(t)\leq \alpha  \}$. Then,
					\begin{equation}
						\hat{t}_n^l \stackrel{p}{\rightarrow}\mathcal{T}_n^l
					\end{equation}
				\end{lemma}
				
				Proof of\textbf{ Lemma \ref{lma9}} can be found in the proof of \textbf{Lemma A.5} in \cite{SC07}. We define 
				\begin{equation}
					\label{dfn11}
					\hat{Q'}_n(t) = \begin{cases}
						\frac{\sum_{j=1}^m\mathbbm{1}(\hat{t}^j_n\geq t)(1-\hat{t}^j_n)}{\sum_{j=1}^m\mathbbm{1}(\hat{t}^j_n\geq t)} & t \in [0,t_n^{(m)}]\\
						0 & t\in(t_n^{(m)},1]
					\end{cases}
				\end{equation}
				
				The function $\hat{Q'}_n(t) $ has the following properties: 
				
				\begin{lemma}
					\label{lma10}
					\begin{enumerate}
						\item 	$\hat{Q'}_n(t)$ is constant in the interval $(\hat{t}_n^{(m-r)},\hat{t}_n^{(m-r+1)}]$, for $r=1(1)m-1$ and in the intervals $[0,\hat{t}_n^{(1)}]$ and $(\hat{t}_n^{(m)},1]$
						\item   $\hat{Q'}_n(t)$ is left-continuous in $(0,1)$.
						\item $\hat{Q'}_n(t)$ is non-increasing in $[0,1]$.
					\end{enumerate}	
				\end{lemma}
				
				The proof is similar to the proof of \textbf{Lemma \ref{lma8}} and so is omitted. It is easy to see that.
				
				\begin{equation}
					\label{dfn10}
					\hat{t}_{n}^u=\inf\{t\in [0,1]: \hat{Q'}_n(t)\leq\beta\}
				\end{equation}
				
				Finally, we find the limiting value of $\hat{t}_{n}^u$ from the following lemma.
				\begin{lemma}
					\label{lma11}
					Define $\mathcal{Q}_n'(t)=\frac{\pi_1 
						\mathbb{P}^1_1(t_n^1\leq t)}{\mathbb{P}_\theta(t_n^1\leq t)\vee1}$ for $t\in[0,1]$.Let,
					$\mathcal{T}_n^u=\inf\{t \in [0,1]: \mathcal{Q}_n'(t)\leq \beta  \}$. Then,
					\begin{equation}
						\hat{t}_n^u \stackrel{p}{\rightarrow}\mathcal{T}_n^u
					\end{equation}
				\end{lemma}
				\begin{corollary}
					\label{cor2}
					Define, $\hat{s}_n=\hat{t}_n^l-\hat{t}_n^u$ and $\mathcal{S}_n=\mathcal{T}_n^l-\mathcal{T}_n^u$. Then,  \begin{equation}
						\hat{s}_n\stackrel{p}{\rightarrow}\mathcal{S}_n \quad \text{as } m\uparrow \infty
					\end{equation}
				\end{corollary}
				The result follows directly from \textbf{Lemmas \ref{lma9}} and \textbf{\ref{lma11}}.
				
				The next lemma ensures a finite stopping time for the data-driven rule even for an indefinitely large number of hypotheses.
				
				\begin{lemma}
					\label{lma12}
					Suppose assumptions \ref{asm1}, \ref{asm2} and \ref{asm3} hold true. Then,
					\begin{equation*}
						\lim_{m\uparrow\infty} \mathbb{P} (T_d<\infty)=1
					\end{equation*}
					and,
					\begin{equation*}
						\lim_{m\uparrow\infty}	\mathbb{P} (T_d\neq n_0)=0
					\end{equation*}
					where, 
					\begin{equation}
						n_0=\inf\{n\in\mathbb{N}: \mathcal{T}_n^l>\mathcal{T}_n^u\}
					\end{equation}
				\end{lemma} 
				
				\begin{corollary}
					\label{cor3}
					Suppose assumptions \ref{asm1} and \ref{asm2} hold true. Then,
					\begin{equation*}
						\lim_{m\uparrow\infty} \mathbb{P} (T<\infty)=1
					\end{equation*}
					and,
					\begin{equation*}
						\lim_{m\uparrow\infty}	\mathbb{P} (T\neq n_0)=0
					\end{equation*}
					where, 
					\begin{equation}
						n_0=\inf\{n\in\mathbb{N}: \mathcal{T}_n^l>\mathcal{T}_n^u\}
					\end{equation}
				\end{corollary}
				Therefore, as the number of hypotheses increases, the oracle stopping time $T$ converges weakly towards a finite natural number $n_0$. Proof of \textbf{corollary \ref{cor3}} is a consequence of the proof of \textbf{Lemma \ref{lma12}}.  
				
				The proof of \textbf{Theorem \ref{thm_3}} is completed by the following lemma.
				
				\begin{lemma}
					\label{lma13}
					Suppose assumptions \ref{asm1}, \ref{asm2} and \ref{asm3} hold. Let $s\in[m-\hat{a}_{T_d},\hat{r}_{T_d}].$ Define $\hat{\mathbf{D'}}=(\hat{D'}^1,\hat{D'}^2,\cdots,\hat{D'}^m)$ where \begin{equation*}
						\hat{D'}^i=\mathbbm{1}(\hat{t}_{T_d}^{i}\leq\hat{t}_{T_d}^{(s)})
					\end{equation*}
					Then, for such a test $(T_d,\hat{\mathbf{D'}})\in \Delta'(\alpha,\beta)$
				\end{lemma} 
				
				
				
			\end{proof}
			\section{Proof of Lemmas} \label{appendixD}
			\begin{proof}[\textbf{Proof of \textbf{Lemma \ref{lma1}} :}]
				First, note that, for continuous test statistics $\mathbf{S}_n$, $\mathbb{P}(t_n^{*i}=1)=\mathbb{P}(t_n^{*i}=0)=0 \ \forall i = 1(1)m $.
				
				$if$ part:
				We assume that, $t_n^{*l}>t_n^{*u}$. For the trivial cases, i.e., when $t_n^{*u}=0$ for example, by definition, $\mathrm{a}_{n}=m$. Therefore,  $\mathrm{r}_{n}+\mathrm{a}_{n}\geq m$ almost surely. The same follows when $t_n^{*l}=1$.
				
				So, if the non-trivial case is true, i.e., $0<t_n^{*u}<t_n^{*l}<1$,
				by definitions of $t_n^{*l}$ and  $t_n^{*u} $, $t_n^{*l}=t_n^{*(\mathrm{r}_{n}+1)}$ and $t_n^{*u}=t_n^{*(m-\mathrm{a}_{n})}$. And finally, $\mathrm{r}_{n}+\mathrm{a}_{n}\geq m$ almost surely. The second part of the lemma is therefore proved.
				
				$only \ if$ part:
				Let $\mathrm{r}_{n}+\mathrm{a}_{n}\geq m$. For $\mathrm{r}_{n}=m$, we know, $t_n^{*l}=1$ almost surely, which in turn proves that $t_n^{*l}>t_n^{*u}$ almost surely. Similarly if $\mathrm{a}_{n}=m$, the same result follows. 
				
				For the non trivial case, i.e., when $\mathrm{r}_{n}<m$ and $ \mathrm{a}_{n}< m$, we have $t_n^{*l}=t_n^{*(\mathrm{r}_{n}+1)}$ and $t_n^{*u}=t_n^{*(m-\mathrm{a}_{n})}$. Finally we have $\mathrm{r}_{n}+\mathrm{a}_{n}\geq m$ i.e., $\mathrm{r}_{n}+1 > m -\mathrm{a}_{n}$, i.e., $t_n^{*l}>t_n^{*u} $ almost surely. 
			\end{proof}
			
			\begin{proof}[\textbf{Proof of Lemma \ref{lma2}}]
				\begin{equation*}
					\begin{split}
						t_n^{*i} = & \mathbb{P}(\theta_i=0|Z_n) \\
						= & \frac{\pi_0f^*_n(Z_n|\theta_i=0)}{\pi_0f^*_n(Z_n|\theta_i=0)+\pi_1f^*_n(Z_n|\theta_i=1)} \text{ (due to (2.4))}\\
						= & \frac{\pi_0f_0(Z_n^i)\prod_{j\neq i} f_n(Z_n^j)}{\pi_0f_0(Z_n^i)\prod_{j\neq i} f_n(Z_n^j)+\pi_1f_{1n}(Z_n^i)\prod_{j\neq i} f_n(Z_n^j)} \\
						= & \frac{\pi_0f_0(Z_n^i)}{f_n(Z_n^i)} = t_n^i
					\end{split}
				\end{equation*}
			\end{proof}
			
			\begin{proof}[\textbf{Proof of Lemma \ref{lma4}}.]

				\textbf{Assumption 2} states that
				\begin{equation*}
					X_n^j|\theta_j \stackrel{iid}{\sim} \theta_j f_1 + (1-\theta_j) f_0, \ n \in \mathbb{N}
				\end{equation*} 
				with \begin{equation*}
					\theta_j \sim \text{Bernoulli}(\pi_1) 
				\end{equation*}
				for each $j\in 1(1)m$.
				
				Then total number of alternative is $K=\sum_{j=1}^m \theta_j$. i.e. $K \sim \text{Bin}(m,\pi_1)$. 
				
				
				Let $T_g^*$ be the stopping time for GAP rule for number of coordinates $m$. i.e.,
				
				\begin{equation}
					T_g^* = \inf\{n\in\mathbb{N} \ | \ \Lambda_n^{(K)}-\Lambda_n^{(K+1)}\geq \log\frac{K(m-K)}{\alpha\wedge\beta}\}
				\end{equation}
				
				Where, $\Lambda_n^{j}$ is the log-likelihood ratio corresponding to the $j$-th coordinate , $ j \in 1(1)m$ and $\Lambda_n^{(1)}\geq\Lambda_n^{(2)}\geq \cdots \geq \Lambda_n^{(m)}$ is correspondingly the ordered representation of the log-likelihood ratios. For any finite positive integer $L$,
				\begin{equation}
					\begin{split}
						\mathbb{P}(T_g^*\leq L) &\leq \mathbb{P}(\cup_n A_n(m))\\
						&\leq \sum_n \mathbb{P}( A_n(m))\\
					\end{split}
				\end{equation}        
				Where \begin{multline*}
					A_n(m)=\{\omega\in \Omega \  | \ \log(S_n^{(K(\omega))}(\omega))-\log(S_n^{(K(\omega)+1)}(\omega))\geq\\ \log\bigg(\frac{K(\omega)(m-K(\omega))}{\alpha\wedge\beta} \bigg)\}
				\end{multline*}.
				
				Now let $\Lambda_n^{j}, \ j=1(1)m$ are iid with cdf $H_n(.)$, 
				\begin{multline*}
					B_{\epsilon,c}(m)=\{\omega\in \Omega : |\Lambda_n^{(Z_m)}(\omega)-H_n^{-1}(\pi_1)|<\epsilon \text{ for } Z_m \in \mathbb{N}, Z_m = m\pi_1 + a_m, \\
					|a_m| < c\sqrt{m}\log(m)\}
				\end{multline*}
				
				Due to \textbf{Lemma 6.} of \cite{Bahadur66}, $\exists M_1 \ (\text{ depending on } c \ \& \ \epsilon) \in \mathbb{N}$ such that $\mathbb{P}(\cap_{m\geq M_1}	B_{\epsilon,c}(m))=1$. 
				
				Therefore, it is evident that, for all $c,\epsilon>0$ $ \mathbb{P}(	B_{\epsilon,c}(m))\rightarrow1$ as $m\uparrow\infty.$
				
				For $c>0$, let
				\begin{equation*}
					D_c(m) = \{\omega\in \Omega : |K(\omega)-m\pi_1| < c\sqrt{m}\log(m)\}
				\end{equation*}
				Due to Bernstein's inequality, it can be shown that, for fixed $c>0$, $\mathbb{P}(D_c(m))\rightarrow 1$ as $m \uparrow \infty$.
				So we have, for fixed $c,\epsilon>0$, 
				\begin{equation}
					\label{eqn99}
					\mathbb{P}(D_c(m)\cap B_{\epsilon,c}(m)) \geq \mathbb{P}(D_c(m))+ \mathbb{P}(	B_{\epsilon,c}(m)) -1
				\end{equation}
				Which converges to 1 as $m \uparrow \infty$.
				
				Now, for any $\omega \in D_c(m)\cap B_{\epsilon,c}(m)$, $ |\Lambda_n^{(K(\omega))}(\omega)-H_n^{-1}(\pi_1)|<\epsilon \ $ \& $\  |\Lambda_n^{(K(\omega)+1)}(\omega)-H_n^{-1}(\pi_1)|<\epsilon$ i.e. $|\Lambda_n^{(K(\omega))}(\omega)-\Lambda_n^{(K(\omega)+1)}(\omega)|<2\epsilon$. But, $\log(K(\omega)(m-K(\omega))/(\alpha\wedge\beta))>2\log(m)+\zeta$ for some sufficiently small $\zeta$.
				
				Therefore,
				\begin{equation*}
					\begin{split}
						\mathbb{P}(A_n(m))= & \mathbb{P}(A_n(m) \cap (D_c(m)\cap B_{\epsilon,c}(m))) + \\ & \mathbb{P}(A_n(m) \cap (D_c(m)\cap B_{\epsilon,c}(m))^c)\\
					\end{split}
				\end{equation*}
				Where the first term converges to 0 due to the discussion in the previous paragraph and the second term converges to 0 because of \ref{eqn99}.
				
				Finally, we observe that $\mathbb{P}(T_g^*>L) \rightarrow 0$ as $m \uparrow \infty$ for any $L \in \mathbb{N}$, which proves the lemma.	
			\end{proof}
			
			\begin{proof}[\textbf{Proof of Lemma \ref{lma5}:}]
				By (2.1), $FDR$ due to the sequential test $(\tau,\mathbf{D}')$ is:
				\begin{equation*}
					\begin{split}
						FDR=&  E( \frac{\sum_{i=1}^m(1-\theta^i)D'^i}{(\sum_{i=1}^m D'^i)\vee1}) \\
						=& E_{Z_{\tau}} (E_{\theta |Z_{\tau}} ( \frac{\sum_{i=1}^m(1-\theta^i)D'^i}{(\sum_{i=1}^m D'^i)\vee1}|Z_{\tau}))\\
						= & E_{Z_{\tau}}  ( \frac{\sum_{i=1}^m(1-E(\theta^i|Z_{\tau}))D'^i}{(\sum_{i=1}^m D'^i)\vee1})\\
						= & E_{Z_{\tau}}  ( \frac{\sum_{i=1}^m t_{\tau}^{*i}D'^i}{(\sum_{i=1}^m D'^i)\vee1})\\
						= & E_{Z_{\tau}}  ( \frac{1}{l\vee 1}\sum_{i=1}^l t_{\tau}^{*(i)})\\
						\leq & E_{Z_{\tau}}  ( \frac{1}{\mathrm{r}_{\tau}\vee 1}\sum_{i=1}^{\mathfrak{r}_{\tau}} t_{\tau}^{*(i)})\\ 
						\leq & \alpha  
					\end{split}
				\end{equation*} 
				The first inequality occurs since $l\in [\mathrm{r}_\tau]$ and the sequence of cumulative average of ordered(increasing) values is increasing in number of terms involved. The second inequality comes from (3.5).
			\end{proof}

			\begin{proof}[\textbf{Proof of Lemma \ref{lma6}:}]
				By (2.2), $FNR$ due to the sequential test $(\tau,\mathbf{d}')$ is:
				\begin{equation*}
					\begin{split}
						FNR=&  E( \frac{\sum_{i=1}^m(1-d'^i)\theta^i}{(\sum_{i=1}^m (1-d'^i))\vee1}) \\
						=& E_{Z_{\tau}} (E_{\theta |Z_{\tau}} ( \frac{\sum_{i=1}^m(1-d'^i)\theta^i}{(\sum_{i=1}^m (1-d'^i))\vee1}|Z_{\tau}))\\
						= & E_{Z_{\tau}} ( \frac{\sum_{i=1}^m(1-d'^i)E_{\theta }(\theta^i|Z_{\tau})}{(\sum_{i=1}^m (1-d'^i))\vee1})\\
						= & E_{Z_{\tau}}  ( \frac{\sum_{i=1}^m (1-t_{\tau}^{*i})(1-d'^i)}{(\sum_{i=1}^m (1-d'^i))\vee1})\\
						= & E_{Z_{\tau}}  ( \frac{1}{l\vee 1}\sum_{i=1}^l (1-t_{\tau}^{*(m-i+1)}))\\
						\leq & E_{Z_{\tau}}  ( \frac{1}{\mathrm{a}_{\tau}\vee 1}\sum_{i=1}^{\mathfrak{a}_{\tau}}  (1-t_{\tau}^{*(m-i+1)}))\\ 
						\leq & \beta  
					\end{split}
				\end{equation*} 
				As in the previous proof, since $l \in [\mathrm{a}_\tau]$ , and due to the fact that the sequence of cumulative average of ordered (increasing) values is increasing in number of terms involved, the first inequality occurs. The second inequality comes from (3.6).
			\end{proof}

			\begin{proof}[\textbf{Proof of Lemma \ref{lma7}.:}]
				
				Define, for fixed $n\in\mathbb{N}$, $\mathcal{A}_n=\{\omega:t_n^{*i}(\omega)\neq t_n^{i}(\omega)\}$.
				Due to \textbf{Lemma \ref{lma2}}, 
				\begin{equation*}
					\mathbb{P}(\mathcal{A}_n)=0
				\end{equation*}
				Define, $\mathcal{B}=\{\omega:t^{*i}_{T(\omega)}(\omega)\neq t^{i}_{T(\omega)}(\omega)\}$.
				Let $\omega_0\in \mathcal{B}$.  Then $T(\omega_0)=n_0\in\mathbb{N}$. Therefore, $\omega_0\in \mathcal{A}_{n_0}$.
				i.e., $\mathcal{B}\subseteq \cup_{n\in \mathbb{N}} \mathcal{A}_n$.
				Finally, \begin{equation}
					\begin{split}
						\mathbb{P}(\mathcal{B})\leq &\mathbb{P}(\cup_{n\in \mathbb{N}} \mathcal{A}_n)\\
						\leq & \sum_{n\in\mathbb{N}} \mathbb{P}( \mathcal{A}_n)\\
						= & 0
					\end{split}
				\end{equation}
				This completes the proof.
			\end{proof}

			\begin{proof}[\textbf{Proof of Lemma \ref{lma8}.}]
				\begin{enumerate}
					
					\item Note that, $\hat{Q}_n(\hat{t}_n^{(r)})=\frac{\sum_{j=1}^m\mathbbm{1}(\hat{t}^j_n\leq \hat{t}_n^{(r)})\hat{t}^j_n}{\sum_{j=1}^m\mathbbm{1}(\hat{t}^j_n\leq \hat{t}_n^{(r)})}$ .
					Now, by definition of $\hat{t}_n^{(r)}$, $\sum_{j=1}^m\mathbbm{1}(\hat{t}^j_n\leq \hat{t}_n^{(r)})\hat{t}^j_n = \sum_{j=1}^r \hat{t}_n^{(j)}$	and $\sum_{j=1}^m\mathbbm{1}(\hat{t}^j_n\leq \hat{t}_n^{(r)})=r$.
					
					So, $\hat{Q}_n(\hat{t}_n^{(r)})=\frac{1}{r}\sum_{j=1}^r \hat{t}_n^{(j)}$.
					Now, for $t \in (\hat{t}_n^{(r)},\hat{t}_n^{(r+1)})$, there is no local FDR value. Hence for such $t$, $\sum_{j=1}^m\mathbbm{1}(\hat{t}^j_n\leq t)\hat{t}^j_n = \sum_{j=1}^r \hat{t}_n^{(j)}$	and $\sum_{j=1}^m\mathbbm{1}(\hat{t}^j_n\leq t)=r$. So, for any $t \in [t_n^{(r)},t_n^{(r+1)}),\ \hat{Q}_n(t)=\frac{1}{r}\sum_{j=1}^r \hat{t}_n^{(j)}$.
					
					Similarly, for $ t \in [\hat{t}_n^{(m)},1],\hat{Q}_n(t)=\frac{1}{m}\sum_{j=1}^m \hat{t}_n^{(j)} $.
					
					By definition, for $ t \in [0,\hat{t}_n^{(1)}),\hat{Q}_n(t)=0 $.
					
					Hence, \textbf{ Lemma \ref{lma8}.} \textit{1} is proved.
					
					\item \textbf{ Lemma \ref{lma8}.} \textit{1} implies that value of $\hat{Q}_n(t)$ is constant in the intervals $(0,\hat{t}_n^{(1)})$; $(\hat{t}_n^{(r)},\hat{t}_n^{(r+1)})$ for $ r=1(1)m-1$ and $(\hat{t}_n^{(m)},1)$. And therefore is continuous in those intervals. Our goal is to show that, $\hat{Q}_n(t)$ is right continuous at the points: \{$\hat{t}_n^{(r)}$ for $r=1(1)m$ \}. Now, from the discussion in the previous part, we can conclude that, for $h < \hat{t}_n^{(r+1)}-\hat{t}_n^{(r)}$, $\hat{Q}_n(\hat{t}_n^{(r)}+h) = \hat{Q}_n(\hat{t}_n^{(r)})$ for $r=1(1)m-1$.
					
					Therefore, $\lim_{h\downarrow 0} \hat{Q}_n(\hat{t}_n^{(r)}+h)=\hat{Q}_n(\hat{t}_n^{(r)})$ and hence, $\hat{Q}_n(t)$ is right continuous in the points \{$\hat{t}_n^{(r)}$ for $r=1(1)m-1$ \}.
					
					The same idea applies for $\hat{t}_n^{(m)}$ for $\hat{Q}_n(\hat{t}_n^{(m)})$ for $h < 1-\hat{t}_n^{(m)}$ and thus $\hat{Q}_n(t)$ is right continuous at the point $\hat{t}_n^{(m)}$.
					
					Hence, \textbf{ Lemma \ref{lma8}.} \textit{2} is proved.
					
					\item From \textbf{ Lemma \ref{lma8}.} \textit{1} we can deduce that, $\hat{Q}_n(t)$ is constant in $[0,1]$ except for jumps at the points \{$\hat{t}_n^{(r)}$ for $r=1(1)m$ \}. To prove that $\hat{Q}_n(t)$ is non decreasing, we are done if we can prove that jumps at the points mentioned above are positive.
					
					Now, jump at point $\hat{t}_n^{(r)}$ for $r=1(1)m-1$ is:
					
					\begin{align*}
						\hat{Q}_n(\hat{t}_n^{(r+1)})-\hat{Q}_n(\hat{t}_n^{(r)})= &  \frac{1}{r+1}\sum_{j=1}^{r+1} \hat{t}_n^{(j)}-\frac{1}{r}\sum_{j=1}^r \hat{t}_n^{(j)}\\
						& = \frac{1}{r(r+1)}{\sum_{j=1}^r} (\hat{t}_n^{(r+1)}-\hat{t}_n^{(1)})>0\\
					\end{align*}
					Hence, \textbf{ Lemma \ref{lma8}.} \textit{3} is proved.
				\end{enumerate}	
			\end{proof}
			
			\begin{proof}[\textbf{Proof of \textbf{Lemma \ref{lma12}.}:}]
				Due to \textbf{Assumption 1.}, as $n\uparrow \infty$, $\mathbb{P}_1^1(t_n^1\geq t)\rightarrow 0 \ \forall t \in (0,1]$ and $\mathbb{P}_0^1(t_n^1\leq t')\rightarrow 0\ \forall t' \in [0,1)$. Therefore, $\forall t \in (0,1)$, $\mathcal{Q}_n(t)\rightarrow 0$ and $\mathcal{Q}_n'(t)\rightarrow 0$ as $n\uparrow\infty$.  And $\mathbb{P}_1^1(t_n^1\geq0 )=\mathbb{P}_0^1(t_n^1\leq1)=1$. So, by definition, $\lim_{n\uparrow\infty}\mathcal{T}_n^l\rightarrow 1$ and $\lim_{n\uparrow\infty}\mathcal{T}_n^u\rightarrow 0$.  
				Now, if for some $n\in\mathbb{N}$, $\hat{s}_n>0$, we must have $T_d\leq n$ almost surely. i.e.,
				\begin{equation*}
					\begin{split}
						\lim_{m\uparrow\infty}\mathbb{P}(T_d\leq n) \geq & \lim_{m\uparrow\infty}\mathbb{P}(\hat{s}_n>0) \\
						\geq & \mathbb{P}(\mathcal{S}_n>\epsilon) \quad \text{for } \epsilon\in(0,1)
					\end{split}
				\end{equation*}
				Last inequality holds due to \textbf{corollary 1.}. As $n\uparrow\infty$, we get, 	\begin{equation*}
					\begin{split}
						\lim_{m\uparrow\infty}\mathbb{P}(T_d\leq\infty) \geq & \lim_{n\uparrow\infty}\mathbb{P}(\mathcal{S}_n>\epsilon) \\
						=& 1
					\end{split}
				\end{equation*}
				Exchange of limits is allowed by BCT since probability values are bounded in $[0,1]$. The last euality holds since $\mathcal{S}_n=\mathcal{T}_n^l-\mathcal{T}_n^u \stackrel{n\uparrow\infty}{\rightarrow} 1( > \epsilon\in(0,1))$. Hence first part is proved.
				
				To prove the second part, note that, by definition of $n_0$, $\mathcal{S}_n\leq0, \ \forall n \in [n_0-1]$ and $\mathcal{S}_{n_0}>0$. Since, $\lim_{n\uparrow\infty}\mathcal{S}_n\rightarrow1$, we must have, $n_0<\infty$.
				Now by \textbf{corollary 1.},  for fixed $n\in\mathbb{N}$, $\hat{s}_n\stackrel{p}{\rightarrow}\mathcal{S}_n$ as $m\uparrow\infty$. So, for $ n\in[n_0]$, for every $\epsilon>0$, for some $\delta>0$, $\exists M_n, $ such that, $\mathbb{P}(|\hat{s}_n-\mathcal{S}_n|<\epsilon)\geq(1-\frac{\delta}{n_0})\ \ \forall m\geq M_n $. Fix $\epsilon<\min\{-\mathcal{S}_1,-\mathcal{S}_2,\cdots,-\mathcal{S}_{n_0-1},\mathcal{S}_{n_0}\}$. For such $\epsilon$, for $m>M_0=\max\{M_1,M_2,\cdots,M_{n_0}\}$,\begin{equation*}
					\mathbb{P}(\hat{s}_n<0)\geq\mathbb{P}(\hat{s}_n<\mathcal{S}_n+\epsilon)\geq\mathbb{P}(|\hat{s}_n-\mathcal{S}_n|<\epsilon)\geq(1-\frac{\delta}{n_0})
				\end{equation*} $\forall n \in[n_0-1]$ and 
				\begin{equation*}
					\mathbb{P}(\hat{s}_{n_0}>0)\geq\mathbb{P}(\hat{s}_{n_0}>\mathcal{S}_{n_0}-\epsilon)\geq\mathbb{P}(|\hat{s}_{n_0}-\mathcal{S}_{n_0}|<\epsilon)\geq(1-\frac{\delta}{n_0}).
				\end{equation*}
				
				Now, for $\delta>0$, $\exists \  M_0 \in \mathbb{N}$  such that $\forall m>M_0$,\begin{equation*}
					\begin{split}
						\mathbb{P}(T_d=n_0)=& \quad	\mathbb{P}(\cap_{n\in[n_0-1]}\{\hat{s}_n<0\}\cap\{\hat{s}_{n_0}>0\}) \\
						\geq &  \sum_{n\in[n_0-1]}\mathbb{P}(\hat{s}_n<0)+\mathbb{P}(\hat{s}_{n_0}>0)-n_0+1\\
						\geq & \quad n_0(1-\frac{\delta}{n_0})-n_0+1\\
						=& \quad 1-\delta
					\end{split}
				\end{equation*}
				
				i.e., $\lim_{m\uparrow\infty} \mathbb{P}(T_d=n_0) =1$, which proves the lemma.
			\end{proof}
			
			\begin{proof}[\textbf{Proof of  Lemma \ref{lma13}}]
				
				For the test $(T_d,\hat{D'})$,
				\begin{equation}
					\begin{split}
						FDR=&  E_\theta( \frac{\sum_{i=1}^m(1-\theta^i)\hat{D'}^i}{(\sum_{i=1}^m \hat{D'}^i)\vee1}) \\
						=& E_{Z_{T_d}} (E_{\theta |Z_{T_d}} ( \frac{\sum_{i=1}^m(1-\theta^i)\hat{D'}^i}{(\sum_{i=1}^m \hat{D'}^i)\vee1}|Z_{T_d}))\\
						= & E_{Z_{T_d}}  ( \frac{\sum_{i=1}^m(1-E(\theta^i|Z_{T_d}))\hat{D'}^i}{(\sum_{i=1}^m \hat{D'}^i)\vee1})\\
						= & E_{Z_{T_d}}  ( \frac{\sum_{i=1}^m t_{T_d}^{*i}\hat{D'}^i_{T_d}}{(\sum_{i=1}^m \hat{D'}^i_{T_d})\vee1})\\
					\end{split}
				\end{equation}
				In the final line we add $T_d$ as the suffix of $\hat{D'}^i$ to emphasis the fact that $\hat{D'}$ depends on the stopping time $T_d$ (which was omitted to maintain simplicity.)
				Now, $T_d=n_0$ implies
				\begin{equation*}
					\frac{\sum_{i=1}^m t_{T_d}^{*i}\hat{D'}^i_{T_d}}{(\sum_{i=1}^m \hat{D'}^i_{T_d})\vee1}=\frac{\sum_{i=1}^m t_{n_0}^{*i}\hat{D'}^i_{n_0}}{(\sum_{i=1}^m \hat{D'}^i_{n_0})\vee1}
				\end{equation*} 
				From \textbf{Lemma \ref{lma12}.} we get,
				\begin{equation*}
					\frac{\sum_{i=1}^m t_{T_d}^{*i}\hat{D'}^i_{T_d}}{(\sum_{i=1}^m \hat{D'}^i_{T_d})\vee1}-\frac{\sum_{i=1}^m t_{n_0}^{*i}\hat{D'}^i_{n_0}}{(\sum_{i=1}^m \hat{D'}^i_{n_0})\vee1}\stackrel{p}{\rightarrow}0
				\end{equation*} 
				as $m\uparrow\infty$. Since the quantity is bounded in $[-1,1]$, we get 
				\begin{equation}
					\label{eqn2}
					\lim_{m\uparrow\infty}( FDR - E(\frac{\sum_{i=1}^m t_{n_0}^{*i}\hat{D'}^i_{n_0}}{(\sum_{i=1}^m \hat{D'}^i_{n_0})\vee1}))=0
				\end{equation}
				So,
				\begin{equation*}
					\lim_{m\uparrow\infty}( FDR-E(\hat{Q}_{n_0}(\hat{t}_{n_0}^{(s)}))=\lim_{m\uparrow\infty}( E(\frac{\sum_{i=1}^m t_{n_0}^{*i}\hat{D'}^i_{n_0}}{(\sum_{i=1}^m \hat{D'}^i_{n_0})\vee1}-\hat{Q}_{n_0}(\hat{t}_{n_0}^{(s)}))
				\end{equation*}
				Now, $\hat{D'}^i_{n_0} = \mathbbm{1}(\hat{t}_{n_0}^{i}\leq\hat{t}_{n_0}^{(s)})$ with, $s\in [m-\hat{a}_{n_0},\hat{r}_{n_0}]$.  And due to \textbf{Assumption 2}, $\mathbb{P}(t_{n_0}^{*i}\neq t_{n_0}^{i})=0$. 
				\begin{equation}
					\begin{split}
						\label{eqn1}
						E\bigg(\frac{\sum_{i=1}^m t_{n_0}^{*i}\hat{D'}^i_{n_0}}{(\sum_{i=1}^m \hat{D'}^i_{n_0})\vee1}-\hat{Q}_{n_0}(\hat{t}_{n_0}^{(s)})\bigg) & = E\bigg(\frac{\sum_{i=1}^m (t_{n_0}^{i}-\hat{t}_{n_0}^{i})\mathbbm{1}(\hat{t}_{n_0}^{i}\leq\hat{t}_{n_0}^{(s)})}{(\sum_{i=1}^m \mathbbm{1}(\hat{t}_{n_0}^{i}\leq\hat{t}_{n_0}^{(s)}))\vee1}\bigg)\\
						&=  E\bigg(\frac{\frac{1}{m}\sum_{i=1}^m (t_{n_0}^{i}-\hat{t}_{n_0}^{i})\mathbbm{1}(\hat{t}_{n_0}^{i}\leq\hat{t}_{n_0}^{(s)})}{\frac{1}{m}(\sum_{i=1}^m \mathbbm{1}(\hat{t}_{n_0}^{i}\leq\hat{t}_{n_0}^{(s)}))\vee1}\bigg)\\
					\end{split}
				\end{equation}
				Now, 
				\begin{multline*}
					var(\frac{1}{m}\sum_{i=1}^m (t_{n_0}^{i}-\hat{t}_{n_0}^{i})\mathbbm{1}(\hat{t}_{n_0}^{i}\leq\hat{t}_{n_0}^{(s)}))  =\frac{1}{m^2}\sum_{i=1}^mvar((t_{n_0}^{i}-\hat{t}_{n_0}^{i})\mathbbm{1}(\hat{t}_{n_0}^{i}\leq\hat{t}_{n_0}^{(s)}))+\\ \frac{1}{m^2}\sum_{i\neq j}^m cov((t_{n_0}^{i}-\hat{t}_{n_0}^{i})\mathbbm{1}(\hat{t}_{n_0}^{i}\leq\hat{t}_{n_0}^{(s)}),(t_{n_0}^{j}-\hat{t}_{n_0}^{j})\mathbbm{1}(\hat{t}_{n_0}^{j}\leq\hat{t}_{n_0}^{(s)}))
				\end{multline*}
				For fixed $i,j\in[m]\ (i\neq j)$ let, 
				\begin{align*}
					\rho_{ij}=& cov((t_{n_0}^{i}-\hat{t}_{n_0}^{i})\mathbbm{1}(\hat{t}_{n_0}^{i}\leq\hat{t}_{n_0}^{(s)}),(t_{n_0}^{j}-\hat{t}_{n_0}^{j})\mathbbm{1}(\hat{t}_{n_0}^{j}\leq\hat{t}_{n_0}^{(s)}))\\
					\leq& var((t_{n_0}^{i}-\hat{t}_{n_0}^{i})\mathbbm{1}(\hat{t}_{n_0}^{i}\leq\hat{t}_{n_0}^{(s)})) \  (= \rho_{ii})\\
					\leq &E(((t_{n_0}^{i}-\hat{t}_{n_0}^{i})\mathbbm{1}(\hat{t}_{n_0}^{i}\leq\hat{t}_{n_0}^{(s)}))^2)\\
					\leq & E(t_{n_0}^{i}-\hat{t}_{n_0}^{i})^2 \rightarrow 0
				\end{align*}
				The last convergence follows from the dominated convergence theorem since $|t_{n_0}^{i}-\hat{t}_{n_0}^{i}| \in [-1,1]$. As a result  	\begin{equation*}
					var(\frac{1}{m}\sum_{i=1}^m (t_{n_0}^{i}-\hat{t}_{n_0}^{i})\mathbbm{1}(\hat{t}_{n_0}^{i}\leq\hat{t}_{n_0}^{(s)}))=\frac{1}{m^2} \bigg[\sum_{i=1}^m \rho_{ii} + \sum_{i\neq j} \rho_{ij}\bigg] \rightarrow 0
				\end{equation*} 
				and therefore due to weak law of large number,
				\begin{equation*}
					\begin{split}
						\frac{1}{m}\sum_{i=1}^m (t_{n_0}^{i}-\hat{t}_{n_0}^{i})\mathbbm{1}(\hat{t}_{n_0}^{i}\leq\hat{t}_{n_0}^{(s)}) &\stackrel{p}{\rightarrow}\ E((t_{n_0}^{i}-\hat{t}_{n_0}^{i})\mathbbm{1}(\hat{t}_{n_0}^{i}\leq\hat{t}_{n_0}^{(s)}))  \rightarrow 0 \\
					\end{split}
				\end{equation*}
				Since, 
				\begin{equation*}
					\begin{split}
						|E((t_{n_0}^{i}-\hat{t}_{n_0}^{i})\mathbbm{1}(\hat{t}_{n_0}^{i}\leq\hat{t}_{n_0}^{(s)})) |&\leq E(|t_{n_0}^{i}-\hat{t}_{n_0}^{i}|)\rightarrow 0.\\
					\end{split}
				\end{equation*}
				Now, if $s=m-\hat{a}_{n_0}$, $\mathbbm{1}(\hat{t}_{n_0}^{i}\leq\hat{t}_{n_0}^{(m-\hat{a}_{n_0})})=\mathbbm{1}(\hat{t}_{n_0}^{i}\leq\hat{t}_{n_0}^{u})$ almost surely due to 4.30.
				\begin{multline}
					\label{rel1}
					var(\frac{1}{m}(\sum_{i=1}^m\mathbbm{1}(\hat{t}_{n_0}^{i}\leq\hat{t}_{n_0}^{u})))=\frac{1}{m^2}\sum_{i=1}^m var( \mathbbm{1}(\hat{t}_{n_0}^{i}\leq\hat{t}_{n_0}^{u})\\
					+\frac{1}{m^2}\sum_{i\neq j}cov( \mathbbm{1}(\hat{t}_{n_0}^{i}\leq\hat{t}_{n_0}^{u}),\mathbbm{1}(\hat{t}_{n_0}^{j}\leq\hat{t}_{n_0}^{u}))
				\end{multline}
				First we consider the convergence of the covariance term. Say,
				\begin{align*}
					\hat{\rho}_{ij}=& cov( \mathbbm{1}(\hat{t}_{n_0}^{i}\leq\hat{t}_{n_0}^{u}),\mathbbm{1}(\hat{t}_{n_0}^{j}\leq\hat{t}_{n_0}^{u}))\\
					= &  \mathbb{P}(\{\hat{t}_{n_0}^{i}-\hat{t}_{n_0}^{u}\leq0\}\cap\{\hat{t}_{n_0}^{j}-\hat{t}_{n_0}^{u}\leq0\})-(\mathbb{P}(\hat{t}_{n_0}^{i}-\hat{t}_{n_0}^{u}\leq0))^2
				\end{align*}
				\textbf{Assumption 3} implies $\hat{t}_{n_0}^{i}-t_{n_0}^{i}\stackrel{p}{\rightarrow}0 \ \forall i \in[m]$ and  \textbf{Lemma \ref{lma11}} implies $\hat{t}_{n_0}^{u} \stackrel{p}{\rightarrow} \mathcal{T}_{n_0}^{u}$ as $m\uparrow\infty$. So,
				\begin{equation*}
					(\hat{t}_{n_0}^{i}-t_{n_0}^{i},\hat{t}_{n_0}^{j}-t_{n_0}^{j},\hat{t}_{n_0}^{u}-\mathcal{T}_{n_0}^{u})\stackrel{p}{\rightarrow}(0,0,0)
				\end{equation*} jointly. 
				Therefore it is easy to see that,
				\begin{multline*}
					\mathbb{P}(\{\hat{t}_{n_0}^{i}-\hat{t}_{n_0}^{u}<0\}\cap\{\hat{t}_{n_0}^{j}-\hat{t}_{n_0}^{u}<0\})
					\rightarrow
					\\ \mathbb{P}(\{t_{n_0}^{i}-\mathcal{T}_{n_0}^{l}<0\}\cap\{t_{n_0}^{j}-\mathcal{T}_{n_0}^{l}<0\})
				\end{multline*} as $m\uparrow \infty$
				Due to \textbf{Assumption 2},
				\begin{equation*}
					\mathbb{P}(\{t_{n_0}^{i}-\mathcal{T}_{n_0}^{u}<0\}\cap\{t_{n_0}^{j}-\mathcal{T}_{n_0}^{u}<0\})=
					(\mathbb{P}(t_{n_0}^{i}-\mathcal{T}_{n_0}^{u}<0))^2 \ \forall \ i,j \in [m]
				\end{equation*}
				and
				\begin{equation*}
					\begin{split}
						\mathbb{P}(\hat{t}_{n_0}^{i}-\hat{t}_{n_0}^{u}<0) \rightarrow \mathbb{P}(t_{n_0}^{i}-\mathcal{T}_{n_0}^{u}<0)
					\end{split}
				\end{equation*}
				So, the covariance term in \ref{rel1} converges to 0 for all $i,j \in[m]$. i.e.,  \begin{equation*}
					\lim_{m\uparrow\infty}  \hat{\rho}_{ij}=0
				\end{equation*}
				And we see that the variance,
				\begin{equation*}
					var(\mathbbm{1}(\hat{t}_{n_0}^{i}<\hat{t}_{n_0}^{l})\leq E(\mathbbm{1}(\hat{t}_{n_0}^{i}<\hat{t}_{n_0}^{l})))\leq 1
				\end{equation*}
				So, 	
				\begin{equation*}
					var(\frac{1}{m}(\sum_{i=1}^m \mathbbm{1}(\hat{t}_{n_0}^{i}<\hat{t}_{n_0}^{l}))) \leq \frac{1}{m^2} \sum_{i=1}^m 1 + \frac{1}{m^2} \sum_{i\neq j} \hat{\rho}_{ij} 	\rightarrow  \frac{1}{m} +0 \rightarrow 0 
				\end{equation*}
				Therefore, due to weak law of large numbers,
				\begin{equation*}
					\frac{1}{m}(\sum_{i=1}^m \mathbbm{u}(\hat{t}_{n_0}^{i}<\hat{t}_{n_0}^{u}))\stackrel{p}{\rightarrow} E(\mathbbm{u}(\hat{t}_{n_0}^{i}<\hat{t}_{n_0}^{u}))\rightarrow \mathbbm{P}_\theta(t_{n_0}^{i}<\mathcal{T}_{n_0}^{u})
				\end{equation*}
				Now, if $s>m-\hat{a}_{n_0}$,
				\begin{equation*}
					\frac{1}{m}\sum_{i=1}^m \mathbbm{1}(\hat{t}_{n_0}^{i}\leq\hat{t}_{n_0}^{(s)}) \geq \frac{1}{m}(\sum_{i=1}^m \mathbbm{1}(\hat{t}_{n_0}^{i}<\hat{t}_{n_0}^{u})) \stackrel{p}{\rightarrow} \mathbb{P}(t_{n_0}^{i}<\mathcal{T}_{n_0}^{u})
				\end{equation*}
				So, for any $s\in[m-\hat{a}_{n_0},\hat{r}_{n_0}]$
				\begin{equation*}
					\begin{split}
						\frac{1}{m}\sum_{i=1}^m (t-\hat{t}_{n_0}^{i}) \mathbbm{1}(\hat{t}_{n_0}^{i}\leq\hat{t}_{n_0}^{(s)})  \stackrel{p}{\rightarrow} 0\\
						\frac{1}{m}\sum_{i=1}^m \mathbbm{1}(\hat{t}_{n_0}^{i}\leq\hat{t}_{n_0}^{(s)}) \stackrel{p}{\rightarrow} \mathbb{P}(t_{n_0}^{i}<\mathcal{T}_{n_0}^{u})>0
					\end{split}
				\end{equation*}
				i.e., due to \ref{eqn1}, for such $s$,
				\begin{equation*}
					\frac{\sum_{i=1}^m t_{n_0}^{*i}\hat{D'}^i_{n_0}}{(\sum_{i=1}^m \hat{D'}^i_{n_0})\vee1}-\hat{Q}_{n_0}(\hat{t}_{n_0}^{(s)}) \stackrel{p}{\rightarrow} 0
				\end{equation*}
				And since,
				\begin{equation*}
					\bigg |\frac{\sum_{i=1}^m t_{n_0}^{*i}\hat{D'}^i_{n_0}}{(\sum_{i=1}^m \hat{D'}^i_{n_0})\vee1}-\hat{Q}_{n_0}(\hat{t}_{n_0}^{(s)}) \bigg | \leq 1
				\end{equation*}
				We conclude, 
				\begin{equation}
					\label{eqn3}
					\begin{split}
						\lim_{m\uparrow\infty} E\bigg(\frac{\sum_{i=1}^m t_{n_0}^{*i}\hat{D'}^i_{n_0}}{(\sum_{i=1}^m \hat{D'}^i_{n_0})\vee1}-\hat{Q}_{n_0}(\hat{t}_{n_0}^{(s)})\bigg) = 0 
					\end{split}
				\end{equation}
				Finally, $ \mathbb{P}(T_d\neq n_0)\rightarrow 1 \implies \hat{a}_{T_d}\stackrel{p}{\rightarrow}\hat{a}_{n_0}$ and $\hat{r}_{T_d}\stackrel{p}{\rightarrow}\hat{r}_{n_0}$. Therefore, for any $s\in [m-\hat{a}_{T_d},\hat{r}_{T_d}]$, $\mathbb{P}(s\in [m-\hat{a}_{n_0},\hat{r}_{n_0}])\rightarrow 1$ and for such $s$,
				\begin{equation}
					\begin{split}
						\lim_{m\uparrow\infty} (FDR-\alpha) = & \lim_{m\uparrow\infty}\bigg[ FDR- E\bigg(\frac{\sum_{i=1}^m t_{n_0}^{*i}\hat{D'}^i_{n_0}}{(\sum_{i=1}^m \hat{D'}^i_{n_0})\vee1}\bigg)+ \\
						&E\bigg(\frac{\sum_{i=1}^m t_{n_0}^{*i}\hat{D'}^i_{n_0}}{(\sum_{i=1}^m \hat{D'}^i_{n_0})\vee1}-\hat{Q}_{n_0}(\hat{t}_{n_0}^{(s)})\bigg)+E\bigg(\hat{Q}_{n_0}(\hat{t}_{n_0}^{(s)})-\alpha\bigg)\bigg]\\
						\leq & \lim_{m\uparrow\infty}\bigg[ FDR- E\bigg(\frac{\sum_{i=1}^m t_{n_0}^{*i}\hat{D'}^i_{n_0}}{(\sum_{i=1}^m \hat{D'}^i_{n_0})\vee1}\bigg)\bigg] +\\
						&  \lim_{m\uparrow\infty} E\bigg(\frac{\sum_{i=1}^m t_{n_0}^{*i}\hat{D'}^i_{n_0}}{(\sum_{i=1}^m \hat{D'}^i_{n_0})\vee1}-\hat{Q}_{n_0}(\hat{t}_{n_0}^{(s)})\bigg)\\
						=& \ 0
					\end{split}
				\end{equation}
				Last inequality follows since, for $s\leq \hat{r}_{n_0}$, $\hat{Q}_{n_0}(\hat{t}_{n_0}^{(s)}) \leq \hat{Q}_{n_0}(\hat{t}_{n_0}^{(\hat{r}_{n_0})})\leq \alpha.$ The limits tend to 0 due to \ref{eqn2} and \ref{eqn3}.
				We therefore observe that, for such a test $(T_d,\hat{D'})$, defined in \textbf{Lemma \ref{lma13}.}, 
				\begin{equation*}
					\lim_{m\uparrow\infty} FDR \leq \alpha
				\end{equation*} 
				We follow similar steps for proving asymptotic control of $FNR$. First note that,
				\begin{equation}
					\begin{split}
						FNR=&  E_\theta\bigg( \frac{\sum_{i=1}^m\theta^i(1-\hat{D'}^i)}{(\sum_{i=1}^m (1-\hat{D'}^i))\vee1}\bigg) \\
						= & E_{Z_{T_d}}  \bigg( \frac{\sum_{i=1}^m (1-t_{T_d}^{*i})(1-\hat{D'}^i_{T_d})}{(\sum_{i=1}^m (1-\hat{D'}^i_{T_d}))\vee1}\bigg)\\
					\end{split}
				\end{equation}
				Due to \textbf{Lemma \ref{lma12}.}, assumption \ref{asm2}. and due to the fact that
				\begin{equation*}
					\bigg|\frac{\sum_{i=1}^m (1-t_{T_d}^{i})(1-\hat{D'}^i_{T_d})}{(\sum_{i=1}^m (1-\hat{D'}^i_{T_d}))\vee1}\bigg| \leq 1,
				\end{equation*}
				we get
				\begin{equation}
					\label{eqn4}
					\lim_{m\uparrow\infty}\bigg[FNR - E_{Z_{T_d}}  \bigg( \frac{\sum_{i=1}^m (1-t_{n_0}^{i})(1-\hat{D'}^i_{n_0})}{(\sum_{i=1}^m (1-\hat{D'}^i_{n_0}))\vee1}\bigg)\bigg]=0
				\end{equation}
				Following our previous proof we can show that for $s\in[m-\hat{a}_{n_0},\hat{r}_{n_0}]$,  
				\begin{equation*}
					\frac{1}{m}\sum_{i=1}^m\mathbbm{1}(\hat{t}_{n_0}^{i}\geq\hat{t}_{n_0}^{(s+1)})(\hat{t}_{n_0}^{i}-t_{n_0}^{i}) \stackrel{p}{\rightarrow} 0
				\end{equation*}
				and 
				\begin{equation*}
					\begin{split}
						&\frac{1}{m}\sum_{i=1}^m\mathbbm{1}(\hat{t}_{n_0}^{i}\geq\hat{t}_{n_0}^{(s+1)})\\
						\geq & \frac{1}{m}\sum_{i=1}^m\mathbbm{1}(\hat{t}_{n_0}^{i}\geq\hat{t}_{n_0}^{l}) \stackrel{p}{\rightarrow} \mathbb{P}(t_{n_0}^{i}\geq\mathcal{T}_{n_0}^{l})>0
					\end{split}
				\end{equation*}
				by weak law of large numbers. Therefore, 
				\begin{equation*}
					\begin{split}
						&\frac{\sum_{i=1}^m (1-t_{n_0}^{i})(1-\hat{D'}^i_{n_0})}{(\sum_{i=1}^m (1-\hat{D'}^i_{n_0}))\vee1}-\hat{Q}_{n_0}'(\hat{t}_{n_0}^{(s)})\\
						=& \frac{\frac{1}{m}\sum_{i=1}^m\mathbbm{1}(\hat{t}_{n_0}^{i}\geq\hat{t}_{n_0}^{(s+1)})(\hat{t}_{n_0}^{i}-t_{n_0}^{i})}{\frac{1}{m}\sum_{i=1}^m\mathbbm{1}(\hat{t}_{n_0}^{i}\geq\hat{t}_{n_0}^{(s+1)})\vee 1} \stackrel{p}{\rightarrow} 0
					\end{split}
				\end{equation*} 
				and since 
				\begin{equation*}
					\bigg| \frac{\frac{1}{m}\sum_{i=1}^m\mathbbm{1}(\hat{t}_{n_0}^{i}\geq\hat{t}_{n_0}^{(s+1)})(\hat{t}_{n_0}^{i}-t_{n_0}^{i})}{\frac{1}{m}\sum_{i=1}^m\mathbbm{1}(\hat{t}_{n_0}^{i}\geq\hat{t}_{n_0}^{(s+1)})\vee 1}\bigg|\leq 1
				\end{equation*}
				we conclude
				\begin{equation}\label{eqn5}
					\lim_{m\uparrow\infty}E\bigg(\frac{\sum_{i=1}^m (1-t_{n_0}^{i})(1-\hat{D'}^i_{n_0})}{(\sum_{i=1}^m (1-\hat{D'}^i_{n_0}))\vee1}-\hat{Q}_{n_0}'(\hat{t}_{n_0}^{(s)})\bigg)=0
				\end{equation}
				Finally, 
				\begin{equation}
					\begin{split}
						\lim_{m\uparrow\infty} (FNR-\beta) = & \lim_{m\uparrow\infty}\bigg[FNR - E_{Z_{T_d}}  \bigg( \frac{\sum_{i=1}^m (1-t_{n_0}^{i})(1-\hat{D'}^i_{n_0})}{(\sum_{i=1}^m (1-\hat{D'}^i_{n_0}))\vee1}\bigg)+ \\
						&E\bigg(\frac{\sum_{i=1}^m (1-t_{n_0}^{i})(1-\hat{D'}^i_{n_0})}{(\sum_{i=1}^m (1-\hat{D'}^i_{n_0}))\vee1}-\hat{Q}_{n_0}'(\hat{t}_{n_0}^{(s)})\bigg)+\\ &E(\hat{Q}_{n_0}'(\hat{t}_{n_0}^{(s+1)})-\beta)\bigg]\\
						\leq & \lim_{m\uparrow\infty}\bigg[ FNR - E_{Z_{T_d}}  \bigg( \frac{\sum_{i=1}^m (1-t_{n_0}^{i})(1-\hat{D'}^i_{n_0})}{(\sum_{i=1}^m (1-\hat{D'}^i_{n_0}))\vee1}\bigg)\bigg] +\\
						&  \lim_{m\uparrow\infty} E\bigg(\frac{\sum_{i=1}^m (1-t_{n_0}^{i})(1-\hat{D'}^i_{n_0})}{(\sum_{i=1}^m (1-\hat{D'}^i_{n_0}))\vee1}-\hat{Q}_{n_0}'(\hat{t}_{n_0}^{(s)})\bigg)\\
						=& \ 0
					\end{split}
				\end{equation}
				Last inequality follows since, for $s\geq m-\hat{a}_{n_0}$, \begin{equation*}
					\hat{Q}_{n_0}'(\hat{t}_{n_0}^{(s+1)}) \leq \hat{Q}_{n_0}(\hat{t}_{n_0}^{(m-\hat{a}_{n_0}+1)}) \leq \beta.
				\end{equation*} The limits tend to 0 due to \ref{eqn4} and \ref{eqn5}.
				We therefore observe that, for such a test $(T_d,\hat{D'})$, defined in \textbf{Lemma \ref{lma13}.}, 
				\begin{equation*}
					\lim_{m\uparrow\infty} FNR \leq \beta
				\end{equation*} 
				So $(T_d,\hat{D'})\in\Delta'(\alpha,\beta)$. i.e., \textbf{Lemma \ref{lma13}.} is proved.
				
			\end{proof}

			 \bibliographystyle{apalike}      

\end{document}